\documentclass[12pt]{article}

\usepackage[english]{babel}
\usepackage{csquotes}
\usepackage{amsfonts,amsmath, amssymb}

\usepackage[OT2,OT1]{fontenc}
\newcommand\cyr{%
  \renewcommand\rmdefault{wncyr}%
  \renewcommand\sfdefault{wncyss}%
  \renewcommand\encodingdefault{OT2}%
  \normalfont
  \selectfont
}
\DeclareTextFontCommand{\textcyr}{\cyr}

\usepackage{hyperref}

\usepackage{xcolor}

\usepackage{mathabx}

\usepackage{enumerate}
\usepackage{epsf,epsfig,amsfonts,a4wide}
\usepackage{amsfonts, mathdots}
\usepackage{amsmath,amssymb,amsthm}
\usepackage{url}
\usepackage{amsmath,amssymb,amsthm}

\usepackage{amssymb}
\usepackage{amsthm}
\usepackage{epsf,epsfig,amsfonts,graphicx}

\newtheorem{Theorem}{Theorem}[section]

\newtheorem{Lemma}{Lemma}[section]
\newtheorem{Proposition}{Proposition}[section]
\newtheorem{Corollary}{Corollary}[section]

\newtheorem{Ex}{Example}[section]
\newtheorem{Remark}{Remark}[section]

\theoremstyle{remark}

\newcommand{\Sym}{\operatorname{Sym}}

\newcommand{\be}{\begin{equation}}
\newcommand{\ee}{\end{equation}}

\newcommand{\R}{\mathbb{R}}\newcommand{\Id}{\textrm{\rm Id}}

\newcommand{\gl}{\mathrm{gl}}

\newcommand{\ddd}{\mathrm{d}}

\newcommand{\pd}[2]{\frac{\partial#1}{\partial#2}}

\newcommand{\dd}{{\mathrm d}\,}

\newcommand{\tr}{\operatorname{tr}}

\newcommand{\trace}{\operatorname{tr}}

\newcommand{\weg}[1]{}

\usepackage[
  backend=biber,
  style=numeric,
  sorting=nyt,
  isbn=false, doi=false, url=false  
]{biblatex}
\DeclareFieldFormat
  [article,inproceedings,incollection,inbook,unpublished]
  {title}{#1}
\title{St\"ackel problem for non-diagonal Killing tensors: Yano-Patterson lifts, algebra of strong symmetries  and  quadratic  in momenta integrals}\author{Alexey V.\ Bolsinov\footnote{School of Mathematics,
 Loughborough University,
 LE11 3TU, UK  and Institute of Mathematics and Mathematical Modeling, Almaty, Kazakhstan\ \ \quad {\tt A.Bolsinov@lboro.ac.uk}},    
Andrey Yu.\  Konyaev\footnote{Faculty of Mechanics and Mathematics and Center for Fundamental and Applied Mathematics, Moscow State University, 119992, Moscow, Russia  and Institute of Mathematics and Mathematical Modeling, Almaty, Kazakhstan
 \ \ \quad{\tt  maodzund@yandex.ru}}  \, and 
    Vladimir S.\ Matveev\footnote{
Institut f\"ur Mathematik, Friedrich Schiller Universit\"at Jena,
07737 Jena,  Germany  \ \ \quad {\tt  vladimir.matveev@uni-jena.de}}}
\date{}

\begin{document}

\maketitle

\begin{abstract}
    We construct integrable Hamiltonian systems  such that  functionally independent Poisson commuting integrals  are  quadratic in the momenta. Unlike the classical St\"ackel setting, we allow the associated self-adjoint   $(1,1)$-tensors $K_\alpha$ to be non-diagonalisable and have Jordan blocks and points where the Segre characteristic changes.
Our construction is  covariant and is based on Nijenhuis geometry: starting from a $\gl$-regular Nijenhuis operator $L$ and its symmetry algebra, we obtain a large class of such integrable systems in a coordinate-free and 
signature-independent way; it is explicit once we have chosen a $\gl$-regular Nijnhuis operator.  In the diagonalisable case, our construction  reproduces the St\"ackel construction, and in dimension $n=2$ it recovers all known systems of this type;   for $n\ge 3$ most of our systems are new. 
Finally, we establish applications to infinite-dimensional integrable systems of hydrodynamic type:  namely, we show that for Killing $(1,1)$-tensors $ K_\alpha$ corresponding to our example  the evolutionarly PDE system  of  hydrodynamic type $u_t = K_\alpha(u)u_x$ is integrable. We  describe its   symmetries, and use generalised reciprocal transformations to reduce it  to a system with  constant coefficient matrices.
\end{abstract}

\tableofcontents
\section{Introduction} \label{sect1}

Let
$g=(g_{ij})$ be an
$n$-dimensional pseudo-Riemannian metric of any signature on a connected open subset  $U\subseteq \mathbb{R}^n(u^1,\dots, u^n)$.
We study the existence of    $n$ functionally independent  quadratic in momenta and Poisson commuting  functions 
\begin{equation}\label{eq:formmetric}
f_\alpha:T^*U\to \mathbb{R}, \ \  f_\alpha(p,u) = \tfrac{1}{2} g(K_\alpha p, p) + V_\alpha(u), \ \  \alpha=1,\dots, n.
\end{equation} 
 Here $K_\alpha=K_\alpha(u)$ are (1,1)-tensor fields self-adjoint with  respect to the metric $g$. We assume $K_1=\Id, $  so that  $f_1$ is the sum of the kinetic energy corresponding to $g$ 
 and a potential $V_1$.  \weg{ In the tensorial notations, the function $F= \tfrac{1}{2} g(Kp, p)+ V$ reads $F=\tfrac{1}{2}g^{sj} K_s^ip_ip_j + V$.}

 The condition that the functions Poisson commute, $\{f_\alpha, f_\beta\}=0$,  implies that the corresponding kinetic parts  also Poisson 
 commute:   $\bigl\{g(K_\alpha p, p),  g(K_\beta p, p)\bigr\}=0$. Therefore,  $K_\alpha$ are Killing  tensors of the metric $g$, that is, each of them satisfies the equation 
 \begin{equation} \label{eq:killin}
     \nabla_k K_{ij} + \nabla_i K_{jk}+ \nabla_j K_{ki}=0, \quad  \mbox{where } K_{ij} = K_i^s g_{sj}.
 \end{equation}

We impose the following additional condition: for all $\alpha,\beta=1,\dots, n$, the $(1,1)$-tensor fields $K_\alpha, K_\beta$ commute as matrices: $K_\alpha K_\beta= K_\beta K_\alpha$ ($\Longleftrightarrow   \   (K_\alpha)^{i}_s (K_\beta)^s_j = (K_\beta)^{i}_s (K_\alpha)^s_j  $).

A well-known example of such a situation comes from the so-called St\"ackel construction\footnote{The construction appeared already in \cite[\S\S 13-14]{Liouville}, see also the discussion in \cite[pp. 703--705]{Luetzen}}, which we recall following  \cite{eisenhart,disser,Bbook}. \weg{\cite{eisenhart,disser,stackel1,stackel2,Bbook}} Take an invertible $n\times n$ matrix $S= (S_{ij})$ with $S_{ij}$ being a
function of the $i$-th variable $u^i$ only. Next, consider the functions $f_\alpha$, $\alpha =1,\dots, n$, given by the following system of linear equations
\begin{equation}
\label{eq:St_intro}
S\mathbb{ F}  = \mathbb{P},
\end{equation}
where  $\mathbb{F}= \left(f_1, f_2,\dots ,f_n\right)^\top$ and $\mathbb{P}= \bigl(p_1^2+ U_1(u^1), p_2^2+ U_2(u^2),\dots ,p_n^2+ U_n(u^ n)\bigr)^\top$. Taking one of them (say, the first one, provided the inverse matrix to $S$ has no zeros in the first row) as 
the Hamiltonian, one obtains a Hamiltonian system with independent Poisson commuting integrals $f_\alpha$ of type \eqref{eq:formmetric} satisfying  the above conditions.

The St\"ackel example played and still plays an important role in mathematical physics and differential geometry. The metric coming from this construction admit separation of variables, which allows one to introduce the Hamilton--Jacobi function via an integral formula and hence find and study the solutions of Hamiltonian equations implicitly, and in certain cases explicitly. A spectacular demonstration of this goes back to Jacobi \cite{Jacobi} and Neumann \cite{Neumann}, who used quadratic in momenta integrals and the corresponding separation of variables to describe geodesics on the ellipsoid and trajectories of the so-called Neumann problem in terms of special functions. In fact, the separation of variables is still the only analytic \weg{(i.e., non-numeric and non-Lie-algebraic)} method to find exact solutions for physically related systems. For example, the book \cite{MS1961} on field theory places particular emphasis on coordinate systems useful for studying various field equations; most of them come from separation of variables. An active research direction initiated in \cite{eisenhart} and continued and finalised in e.g. \cite{Kalninsbook, Kress_book, MaBl2015, separation} is the description of coordinate systems for constant curvature spaces in which the metric takes the St\"ackel form.

In differential geometry, Killing tensors appear naturally in various interesting situations. For example, geodesic equivalence of metrics, which was studied by Beltrami, Levi-Civita and other classics, generates Poisson commuting quadratic in momenta integrals \weg{\cite{MT1997, BM2003}} \cite{MT1997}. These integrals turned out to be an effective tool in this theory and allowed one to solve classical conjectures \cite{BMR2021,BMM2008, Matveev2012}.

Killing tensors were also studied within the framework of general relativity. A spectacular result was the description of geodesics of the so-called Kerr solution  \cite{carter1968,Penrose1970}. Later, many results on the search for and applications of Killing tensors in general relativity and celestial mechanics appeared, e.g. \cite{Woodhouse1975,bagrov, Osetrin2016, Obukhov2021, DRW2004}, to cite a few. Similar to the Kerr solution, the metrics studied in these papers are not genuine St\"ackel metrics, but can be reduced to the St\"ackel form after quotienting with respect to a subgroup of the isometry group.

In these papers and also in the St\"ackel construction, there locally exist coordinate systems in which both the metric and Killing tensor are given by diagonal matrices. Moreover, under the assumption that Killing tensors $K_\alpha$ are linearly independent and the corresponding integrals $f_\alpha$ commute, the simultaneously diagonalisability of $K_\alpha$'s on every cotangent fiber implies simultaneously diagonalisability in a suitable coordinate system, which automatically leads to the St\"ackel construction, see \cite{KM1980, Kiyohara}. In contrast to the Riemannian signature, in the pseudo-Riemannian setting there may be linear-algebraic restrictions on the diagonalisation of self-adjoint $(1,1)$-tensors, since they  may have nontrivial Jordan blocks; the latter case was completely missing in the references above.

The study of non-diagonalisable Killing tensors has of course attracted the attention of many mathematicians and physicists, since it would be natural to try to extend a well-known theory with such impressive applications to a more general case of obvious interest for mathematical physics, in particular for general relativity, see e.g. \cite{Hauser1976,Coll2006, Kokkinos2023}.  Integrable geodesic flows of form \eqref{eq:formmetric} with non-diagonalisable Killing tensors $K_\alpha$ have also naturally appeared in the theory of geodesic equivalence:  the local description of geodesically equivalent metrics finalised in \cite{BM2015} provides examples with Killing $(1,1)$-tensors of arbitrary Segre characteristic. \weg{Next, by \cite{BM2003,MT1997}, geodesically equivalent metrics generate Killing tensors by a certain algebraic formula preserving the Segre characteristics. Combining these results, we obtain an integrable system of the form \eqref{eq:formmetric} such that the Killing $(1,1)$-tensors $K_\alpha$ have arbitrary Segre characteristics. Note though that Killing tensors coming from geodesic equivalence do not provide all examples even in the diagonalisable case, as the St\"ackel construction has the freedom of choice of $n^2$ functions of one variable, whereas nontrivial geodesic equivalence allows the choice of at most $n$ functions of one variable only. Moreover, in relation to general relativity, by \cite{kiosak2009}, four-dimensional Einstein metrics of nonconstant curvature do not admit nontrivial geodesic equivalence.}

In this paper we present an explicit construction of systems of the form \eqref{eq:formmetric} satisfying the above conditions (see Section 4 and Theorem \ref{t2'} therein). The construction works in all signatures and allows any Segre characteristic. It also works at singular points where the Segre characteristic changes  (i.e., eigenvalues collide). The freedom of the construction is, as in the St\"ackel example, $n^2$ functions of one variable.
In dimension $n=2$, the construction provides all possible systems  of type \eqref{eq:formmetric}. In dimensions $n\ge 3$, most of the (non-diagonal) systems coming from our construction are new. 

\weg{As explained above, the problem we solved is natural and has been studied by many people.} 
Let us explain the difficulties appearing in the non-diagonal case. Besides some evident algebraic issues related to the fact that non-diagonal tensors have more components, and the components sitting in different places play different roles, there is one more important problem. If we assume that the metric and its Killing tensor are both diagonal in a certain coordinate system, then the Killing equation \eqref{eq:killin} can be solved with respect to all derivatives; in particular, the entries of a Killing tensor at one point determine this tensor at all points \cite{eisenhart,KKM2024}. Moreover, if one Killing tensor, diagonal in these coordinates, has $n$ different eigenvalues, then one can construct $n-2$ additional independent Killing tensors such that the corresponding quadratic in momenta integrals Poisson commute. Examples, even with a flat background metric, demonstrate that this is not true anymore for other Segre characteristics. In particular, in the non-diagonal case, the PDE system on the components of $g$ and $K_\alpha$ is less overdetermined and, therefore, much more difficult to solve.

A new circle of methods has come from Nijenhuis geometry research program initiated in \weg{\cite{nij1,BMMT2018}} \cite{nij1}. The background structure for our construction is a Nijenhuis operator\footnote{Nijenhuis operators are $(1,1)$-tensor fields whose Nijenhuis torsion vanishes, see \cite{nij1} or Section \ref{sect2}} $L$ on $U$. Other principal ingredients are the symmetry algebra $\Sym L$ and a basis $M_1,\dots, M_n$ of this algebra.
We emphasize that the approach via Nijenhuis operators gives a general invariant setup. If we rewrite the construction in these terms, the algebraic type of operators becomes completely irrelevant, and the difference between diagonal and non-diagonal Killing tensors disappears. Moreover, one can work in an arbitrary coordinate system, so even in the diagonalisable case there is no need to pre-diagonalise the Killing tensors.

In the case when $L$ is $\R$-diagonalisable and has a simple spectrum, for instance, $L = \operatorname{diag}(u_1, \dots, u_n)$, our construction is essentially the same as St\"ackel's method. However, in general, the underlying Nijenhuis structure may have arbitrary algebraic type. In particular, $L$ may admit Jordan blocks and points where the Segre characteristic changes. The only algebraic condition on $L$ is the $\gl$-regularity, that is, the operators $\Id, L, L^2, \dots, L^{n-1}$ must be linearly independent at each point, or equivalently, every eigenvalue must have geometric multiplicity one, whereas the algebraic multiplicity can be arbitrary. 

The notion of symmetries of $(1,1)$-tensor fields will be recalled in Section \ref{sect2}. This concept, including the terminology, was actively studied in the theory of infinite-dimensional integrable systems. Note that there are various relations between finite-dimensional integrable systems of type  \eqref{eq:formmetric} and infinite-dimensional integrable systems. One of the most spectacular was observed and used in \cite{Dubrovin75,Dubrovin1975} and later popularized in, e.g., \cite{moser80}: the so-called finite gap KdV  solutions come from certain integrable Hamiltonian systems of type \eqref{eq:formmetric}. \weg{Recent publications in this direction include \cite{BS2023, BKM2025}.}  An important connection between integrable weakly non-linear systems of hydrodynamic type and integrable systems admitting separation of variables was discovered in  \weg{\cite{Ferapontov1991, Ferapontov1991b, FerapontovFordy1997}} \cite{Ferapontov1991}.    
Of key importance for our result was the interaction between Nijenhuis geometry and infinite-dimensional integrable systems of hydrodynamic type discussed in \cite{ml,nij3,nij4}.

The following three new observations have been crucial for our construction. The first one is that strong symmetries of a Nijenhuis operator  $L$ form an algebra (the product of strong symmetries is a strong symmetry). The second is that these symmetries admit natural (Yano-Patterson, \cite{yano}) cotangent lift and, as a result, naturally interact with the Poisson--Nijenhuis structure  on $T^*\mathsf M$ naturally associated with $L$. And the third one is a construction of a special, in some sense canonical, symmetry $P$ of the lifted Nijenhuis operator $\widehat L$, which is linear in momenta and finally leads to quadratic-in-momenta integrals.

In fact, our  construction provides a more general class of integrable systems than that of type \eqref{eq:formmetric}. In particular, it  naturally extends  the so-called ``generalised St\"ackel systems'' from \cite[Chapter 5]{AKN2006} and ``generalised St\"ackel separation relations'' from \cite[\S 4.2.2]{Bbook}, see Example \ref{ex:4.1},  to the non-diagonal case\weg{Note that bi- and multihamiltonian structures are widely used in the theory of finite- and infinite-dimensional integrable systems and are a separate object of study, see e.g. \cite{BMR2017,FP2003, BMMT2018, IMM,MagriSchw, Magri1995, DKV2016}. As was recently shown in \cite{nijapp2, nijapp3}, Nijenhuis geometry is well adapted for the study of multi-Hamiltonian structures in the infinite-dimensional case. The results of the present paper demonstrate that the objects natural for Nijenhuis geometry, i.e., covariantly defined by a Nijenhuis operator $L$ on $U$, such as symmetries and conservation laws of $L$, are directly connected to compatible Poisson structures on $T^*U$.}

As discussed above, the objects typically considered in the theory of infinite-dimensional integrable systems play a key role in our finite-dimensional construction. The results of this paper, however, also have natural applications in the theory of infinite-dimensional integrable systems, more precisely, quasilinear PDE systems of the form
$
u_{t}= Au_x.
$
Here, $u= (u^1(x,t),\dots, u^n(x,t) )^\top$ is an unknown vector function, and $A=A(u)$ is a matrix depending on $u$ and viewed  as a $(1,1)$-tensor field. We show that if $A$ is a Killing (1,1)-tensor coming from our construction (that is, one of $K_\alpha$ from \eqref{eq:formmetric}), then the corresponding system of hydrodynamic type is integrable, and in particular, for almost every initial curve $u(x,0)$ one can find a solution $u(x,t)$ in quadratures. Moreover, we show that the operators $K_\alpha$ are symmetries  of each other, so that the PDE system $ u_{t_\alpha} = K_\alpha(u) u_{t_1}$, where now
$u = \bigl(u^1(t_1,t_2,\dots,t_n),\dots, u^n(t_1,t_2,\dots,t_n)\bigr)^\top$, satisfies the compatibility conditions and, therefore, is consistent). Furthermore, we show that in a neighbourhood of almost every point, by a suitable generalised reciprocal transformation (see details in Section \ref{sect6}), the equation can be linearised, i.e., all the matrices $K_\alpha$ can be made constant (Corollary \ref{c2}). In the diagonal case, this result follows from  \weg{\cite{Ferapontov1991, Ferapontov1991b, FerapontovFordy199}} \cite{Ferapontov1991}. We emphasize that our approach also works at singular points where the Segre characteristic of $K_\alpha$ changes. In this situation, the above PDE system is not linearisable, but one can still  bring $K_\alpha$ to a ``nicer'' form using the techniques of companion forms developed in \cite{nij3}.

\weg{Recall that the theory of diagonalisable integrable systems of hydrodynamic type is well understood \cite{Tsarev1990}, whereas the nondiagonal case is still very far from being understood; see, e.g., \cite{Ferapontov3,Pavlov2018,  LPV2024}.}

The paper is organised as follows.  In Section \ref{sect2}, we recall the notion of symmetries and strong symmetries. The main result of the section is Theorem \ref{t1}, which describes the behaviour of symmetries and strong symmetries of $(1,1)$-tensor fields under the Yano-Patterson cotangent lift.
\weg{ of strong symmetries to the cotangent bundle is a strong symmetry, which, in turn, implies (Corollary \ref{cor2.1}) a construction of compatible Poisson structures on the cotangent bundle.} In Section \ref{sect3}, we recall necessary results from \cite{nij4} and treat two important examples. Theorem \ref{t2} from Section \ref{sect4} gives a general construction of Poisson commuting functions on a Poisson manifold $(\mathsf M, \mathcal P)$ via symmetries of a Nijenhuis operator $N$ that is self-adjoint with respect to $\mathcal P$. Applying this construction to the Yano-Paterson lift $N=\widehat L$ of our initial Nijenhuis operator $L$, we obtain Poisson commuting functions $f_\alpha$ of type \eqref{eq:formmetric} (Theorem \ref{t2'}). We illustrate this construction by two examples. Theorem \ref{t3} of the next Section \ref{sect5} provides a more detailed analysis of  Poisson commuting functions \eqref{eq:formmetric} associated with $L$ and also provides the solution of the Hamilton--Jacobi equation for this choice of integrals. We also give examples explaining that our formulas give all possible integrable systems of type \eqref{eq:formmetric} in dimension 2, and an example with a flat background metric of Lorentz signature, which demonstrates some advantages of the invariant approach. Section \ref{sect6} discusses the relation of the previous results to integrable systems of hydrodynamic type and reciprocal transformations.


\section{Symmetries, strong symmetries and lifts of operator fields}\label{sect2}

First, we recall some basic definitions from Nijenhuis geometry (see \cite{nijenhuis, nij1, nij4}). Let $L$ be a tensor field of type $(1, 1)$ (operator field) on a manifold $\mathsf{M}^n$. The Nijenhuis torsion of $L$ is the tensor field of type $(1, 2)$ defined by
$$
{{\mathcal N}_L}(\xi, \eta) = L^2 [\xi, \eta] + [L\xi, L\eta] - L [L\xi, \eta] - L [\xi, L\eta],    
$$
where $\xi, \eta$ are vector fields, and $[\, , \,]$ denotes the standard Lie bracket of vector fields. An operator field $L$ is said to be a {\it Nijenhuis operator} if its Nijenhuis torsion vanishes.

Let $L, M$ be a pair of operator fields and set (see \cite{nijenhuis}, formula 3.9)
\begin{equation}\label{ii:eq2}
\langle L, M \rangle (\xi, \eta) = LM[\xi, \eta] + [L\xi, M\eta] - L[\xi, M\eta] - M[L\xi, \eta].
\end{equation}
It is a simple exercise to check that the right-hand side of this formula defines a tensor field of type $(1,2)$ if and only if $LM - ML = 0$. Obviously, $\langle L, L \rangle = \mathcal N_L$. Notice that, in general, the tensor field given by \eqref{ii:eq2} is neither symmetric nor skew-symmetric in $\xi$ and $\eta$. 

For our purposes, we will need another version of formula \eqref{ii:eq2}:
\begin{equation}\label{ii:eq2.3}
 \langle L, M \rangle (\xi, \eta) = \Big(\mathcal L_{L\xi} M - L \mathcal L_\xi M\Big) \eta = - \Big(\mathcal L_{M\eta} L - M \mathcal L_\eta L \Big)\xi, 
\end{equation}
where $\mathcal L_\xi$ stands for the Lie derivative.
 
We say that $L$ and $M$ are {\it symmetries} of each other if $LM=ML$ and the symmetric (in lower indices) part of the tensor $\langle L, M \rangle$ vanishes; that is,  $\langle L, M \rangle(\xi, \xi) = 0$ for all vector fields $\xi$. If the entire tensor $\langle L, M \rangle$ vanishes, then we say that $L$ and $M$ are {\it strong symmetries} of each other. For example, $L$ is always a symmetry of itself, but it is a strong symmetry if and only if $L$ is Nijenhuis.

Let $M$ and $R$ be strong symmetries of $L$. In general, they are not strong symmetries of each other (see Example \ref{ex2.1} for $n > 1$). However, $MR$ is a strong symmetry of $L$ (which follows from the third statement of Lemma \ref{lem2.1}). Thus, strong symmetries of $L$ form an associative algebra with respect to matrix multiplication, which will be denoted by $\Sym L$. 

Next, the formula \cite{frolicher1, frolicher2}
\begin{equation} \label{eq:1}
    \begin{aligned}{}
    [[L, M]]_{FN}(\xi, \eta) & = [L\xi, M\eta] + [M\xi, L\eta] - M[L\xi, \eta] - M[\xi, L\eta] - \\
    & - L[M\xi, \eta] - L[\xi, M\eta] + ML[\xi, \eta] + LM[\xi, \eta].
    \end{aligned}
\end{equation}
defines a skew-symmetric tensor of type $(1, 2)$ known as a {\it Fr\"olicher-Nijenhuis bracket} of $M$ and $L$. Unlike $\langle L, M\rangle$ defined by \eqref{ii:eq2}, $[[L, M]]_{FN}$ is a tensor regardless of the commutativity of $L$ and $M$. At the same time,  for commuting $L, M$,  formula  \eqref{eq:1} reads 
$$
[[L, M]]_{FN} (\xi, \eta) = \langle L, M \rangle (\xi, \eta) - \langle L, M \rangle (\eta, \xi).
$$
In other words, the Fr\"olicher-Nijenhuis bracket coincides with the skew-symmetric part of $\langle L, M \rangle$. Thus, the operators  $L$ and $M$ are strong symmetries of each other if and only if they are symmetries of each other and their Fr\"olicher-Nijenhuis bracket vanishes. In particular, $L$ is Nijenhuis if and only if $[[L, L]]_{FN} = 0$.

\begin{Ex}[Strong symmetries of a complex structure]\label{ex2.1}
\rm{
Let $\mathsf J$ be a complex structure on a real manifold $\mathsf{M}^{2n}$. In canonical coordinates $u^1, \dots, u^n, v^1, \dots, v^n$, it takes the form
$$
\mathsf J = \left( \begin{array}{cc}
     0_{n \times n} & - \Id_{n \times n}  \\
     \Id_{n \times n} & 0_{n \times n}
\end{array}\right),
$$
where $0_{n \times n}$ is the $n \times n$ zero matrix and $\Id_{n \times n}$ is the $n \times n$ identity matrix. Let  $\xi_i, \eta_j$ denote basis vector fields so that
$\mathsf J \xi_i = \eta_i$, $\mathsf J\eta_i = - \xi_i$.

The condition $L\mathsf J - \mathsf J L = 0$ yields 
\begin{equation}\label{3.3:jform}
L = \left(\begin{array}{cc}
     A & -B  \\
     B & A
\end{array}\right),    
\end{equation}
where $A, B$ are $n \times n$ submatrices. The condition $\langle L, \mathsf J \rangle = 0$ splits into four conditions
$$
\begin{aligned}
\langle L, \mathsf J \rangle (\xi_i, \xi_j) & = [L\xi_i, \eta_j] - \mathsf J[L\xi_i, \xi_j] = [A_i^s \xi_s + B^m_i \eta_m, \eta_j] - \mathsf J[A_i^s \xi_s + B^m_i \eta_m, \xi_j] = \\
& = - \Bigg(\pd{A^s_i}{v^j} + \pd{B^s_i}{u^j} \Bigg)\xi_s - \Bigg( \pd{B^m_i}{v^j} - \pd{A^m_i}{u^j}\Bigg)\eta_m, \\
\langle L, \mathsf J \rangle (\xi_i, \eta_j) & = - [L\xi_i, \xi_j] - \mathsf J[L\xi_i, \eta_j] = - \mathsf J ([L\xi_i, \eta_j] - \mathsf J[L\xi_i, \xi_j]) = - \mathsf J \langle L, \mathsf J \rangle (\xi_i, \xi_j), \\
\langle L, \mathsf J \rangle (\eta_i, \xi_j) & = [L\eta_i, \eta_j] - \mathsf J[L\eta_i, \xi_j] = [- B^s_i \xi_s + A^m_i \eta_m, \eta_j] - \mathsf J[- B^s_i \xi_s + A^m_i \eta_m, \xi_j] = \\
& = \Bigg( \pd{B^s_i}{v^j} - \pd{A^s_i}{u^j}\Bigg) \xi_s - \Bigg( \pd{A^m_i}{v^j} + \pd{B^m_i}{u^j}\Bigg) \eta_m, \\
\langle L, \mathsf J \rangle (\eta_i, \eta_j) & = - [L\eta_i, \xi_j] - \mathsf J[L\eta_i, \eta_j] = - \mathsf J  \langle L, \mathsf J \rangle (\eta_i, \xi_j).
\end{aligned}
$$
We see that the brackets contain exactly Cauchy-Riemann conditions for the components of $L$. Thus,  the strong symmetries of a complex structure are exactly the holomorphic operator fields. If $n = 1$, the algebra $\Sym L$ is isomorphic to the algebra of complex holomorphic functions on $\mathsf M^2$ as a complex 1-manifold. 
}\end{Ex}

Let $\theta$ be the Liouville form on the cotangent bundle $T^*\mathsf{M}^n$ of a manifold $\mathsf{M}^n$. In canonical coordinates $p_1, \dots, p_n, u^1, \dots, u^n$,  it takes the form $\theta = p_q \ddd u^q$ (here $u^i$ are coordinates on $\mathsf{M}^n$ and $p_i$ are the corresponding momenta). For any operator field $L$ on $\mathsf{M}^n$, the {\it deformed} Liouville form is defined by $\theta_L = p_q L^q_s \ddd u^s$. Obviously, $\theta_{\Id} = \theta$.

Consider the canonical symplectic form $\Omega = \ddd \theta$ on $\mathsf{M}^n$ and another 2-form $\Omega_L = \ddd \theta_L$ constructed from the {\it deformed} Liouville form. In canonical coordinates, we obtain
$$
\Omega = \ddd p_q \wedge \ddd u^q, \quad \Omega_L = p_q \pd{L^q_s}{u^r} \,\ddd u^r \wedge \ddd u^s + L^q_s\, \ddd p_q \wedge \ddd u^s.
$$
The {\it complete lift} of $L$ onto the cotangent bundle is the operator field $\widehat L$ on $T^*\mathsf M$, which is uniquely defined by the relation \cite{yano}  
\begin{equation}\label{2.7:lift}
\Omega (\cdot, \widehat L \cdot ) = \Omega_L(\cdot, \cdot).     
\end{equation}
In canonical coordinates, $\widehat L$ takes the form
\begin{equation}
\label{eq:Llift}
\widehat L = \Omega^{-1} \Omega_L = \left(\begin{array}{cc}
     L^{\top} & R  \\
     0_{n \times n} & L 
\end{array}\right), \quad \text{where} \quad R_{ij} = \Big(\pd{L^q_i}{u^j} - \pd{L^q_j}{u^i}\Big) p_q.      
\end{equation}
From these formulas, it immediately follows that $\Omega_L^{-1}$ defines a Poisson bracket on $T^*\mathsf{M}^n$ as soon as  
$\det L \neq 0$.

\begin{Remark}\label{rem2.2}
\rm{
Let $L$ be a Nijenhuis operator in a neighbourhood of point $\mathsf{p}$ and assume that in coordinates $u^1, \dots, u^n$ we have
\begin{equation}\label{cond}
\ddd (L^* \ddd u^i) = 0.  
\end{equation}
Locally there exist such functions $h^i$ that $\ddd h^i = L^* \ddd u^i$ and, thus
$$
\theta_L = p_q L^q_s \ddd u^s = p_q \ddd h^q.
$$
The complete lift of $L$ in these coordinates takes the form
\begin{equation}\label{block}
\widehat L = \left(\begin{array}{cc}
     L^{\top} & 0_{n\times n}  \\
     0_{n\times n} & L 
\end{array}\right).    
\end{equation}
Conversely, if $\widehat L$ is given by \eqref{block},  then $R^i_j = \pd{L^q_i}{u^j} - \pd{L^q_j}{u^i} = 0$ (see general formula \eqref{eq:Llift}), then
the rows of $L$, which can be understood as $1$-forms $L^* \ddd u^i$, are all closed. Thus, vanishing of the $R$-block in \eqref{eq:Llift} is equivalent to condition \eqref{cond}.

}    
\end{Remark}

Our first theorem explains what happens to the Fr\"olicher-Nijenhuis bracket, symmetries and strong symmetries under complete lifts.

\begin{Theorem}\label{t1}
Consider two operator fields $L, M$ on  $\mathsf{M}^n$ and their lifts $\widehat L, \widehat M$ on $T^*\mathsf M$.
Then:
\begin{enumerate}
    \item If the Fr\"olicher-Nijenhuis bracket of $L$ and $M$ on $\mathsf{M}^n$ vanishes, then the Fr\"olicher-Nijenhuis bracket of $\widehat L$ and $\widehat M$ vanishes on $T^*\mathsf{M}^n$.
    \item If $M$ and $L$ are symmetries of each other, then $\widehat M$ and $\widehat L$ are symmetries of each other.
    \item If $M$ and $L$ are strong symmetries of each other, then $\widehat M$ and $\widehat L$ are strong symmetries of each other. Moreover, $\widehat{LM} = \widehat L \widehat M$.
\end{enumerate}    
\end{Theorem}

\begin{Corollary}\label{cor2.1}
Let $L, M$ be Nijenhuis operators on $\mathsf{M}^n$, and $\det M \neq 0, \det L \neq 0$. If  $[[L, M]]_{FN} = 0$, then the Poisson structures $\Omega^{-1}_{L^{-1}}$ and $\Omega^{-1}_{M^{-1}}$ are compatible.
\end{Corollary}

\begin{Remark}\label{rem2.1}
\rm{
A vector subspace $\mathcal N$ in the space of operator fields on $\mathsf{M}^n$ is called a {\it Nijenhuis pencil}, if it consists of Nijenhuis operators only (see \cite{nijapp2, Konyaev2023}). Equivalently,  for all $L, M \in \mathcal N$ (not necessarily different), we have $[[L, M]]_{FN} = 0$. Corollary \ref{cor2.1} implies that every Nijenhuis pencil on $\mathsf M$ induces a natural Poisson pencil on the cotangent bundle $T^*\mathsf M$. 
}     
\end{Remark}

\begin{proof}[Proof of Theorem \ref{t1}]
All statements can, of course, be checked by gruesome computation based on the explicit formula \eqref{eq:Llift} for the lift $\widehat L$. We, however, prefer an invariant approach and start with the following Lemma, which is useful on it own.

\begin{Lemma}\label{lem2.1}
For arbitrary operators $L, M, R$, where $M$ and $R$ commute with $L$ (but not necessary with each other) and a smooth function $f$, the following tensor relations hold:
    \begin{equation}\label{ii:tensor}
    \begin{aligned}
       & \langle f L, M \rangle = f \langle L, M \rangle + M L \otimes \ddd f - L \otimes M^* \ddd f, \\
       & \langle L, f M \rangle = f \langle L, M \rangle + L^* \ddd f \otimes M - \ddd f \otimes ML \\
    \end{aligned}
    \end{equation}
    \begin{equation}\label{ii:eq3}
    \langle L, MR \rangle (\xi, \eta) = \langle L, M \rangle (\xi, R \eta) + M \langle L, R \rangle (\xi, \eta).    
    \end{equation}
\end{Lemma}
\begin{proof}
By  direct computation we have
$$
\begin{aligned}
 \langle f L, M \rangle (\xi, \eta) & = f LM[\xi, \eta] + [f L\xi, M\eta] - f L[\xi, M\eta] - M[f L\xi, \eta] = \\
 & = f \langle  L, M \rangle(\xi, \eta) + \mathcal L_\eta f \cdot ML \xi - \mathcal L_{M\eta} f \cdot L\xi = \\
 & = (\langle L, M \rangle + ML \otimes \ddd f - L \otimes M^*\ddd f) (\xi, \eta).
\end{aligned}
$$
The case $\langle L, fM \rangle$ is treated in a similar way. To prove \eqref{ii:eq3},  we write:
$$
\begin{aligned}
& \langle L, M \rangle (\xi, R\eta) = LM[\xi, R\eta] + [L\xi, MR\eta] - L[\xi, MR\eta] - M[L\xi, R\eta], \\
&  M \langle L, R \rangle (\xi, \eta) = MLR[\xi, \eta] + M[L\xi, R\eta] - ML[\xi, R\eta] - MR[L\xi, \eta], \\
& \langle L, MR \rangle (\xi, \eta) = LMR[\xi, \eta] + [L\xi, MR\eta] - L[\xi, MR\eta] - MR[L\xi, \eta].
\end{aligned}
$$
Subtracting the first and second relations from the third one and keeping in mind that $M$ commutes with $L$, we get zero in the right hand side,  which is equivalent to \eqref{ii:eq3}.
\end{proof}

Now let us proceed to the proof of Theorem. First, fix canonical coordinates $p_1, \dots, p_n, u^1, \dots, u^n$ on $T^* \mathsf{M}^n$. We will need two more definitions, both due to \cite{yano}, concerning the lift of tensor fields onto the cotangent bundle.

For any vector field $\xi$ on $\mathsf{M}^n$, one can define a function $h_\xi = p_q \xi^q$ on the cotangent bundle $T^* \mathsf{M}^n$. The {\it complete lift}  $\widehat \xi$ of the vector field $\xi$ is  
$$
\widehat \xi = \Omega^{-1} \ddd h_\xi.
$$
Equivalently, we can characterise $\widehat \xi$ as the Hamiltonian vector field generated by $h_\xi$.
In canonical coordinates coordinates (there is a summation in repeating indices)
$$
\widehat \xi = - p_q\pd{\xi^q}{u^i} \pd{}{p_i} + \xi^s \pd{}{u^s}.
$$
It is easy to see that at every point $(p, u)\in T^*\mathsf M$, $p\ne 0$, the lifts $\widehat \xi$ span $T_{(p,u)}(T^* \mathsf{M}^n)$ and
$$
[\widehat \xi, \widehat \eta] = \widehat{[\xi, \eta]} \quad   \mbox{for all $\xi$, $\eta$.}
$$

Now consider an operator field $L$ on $\mathsf{M}^n$. The {\it vertical lift} $\widetilde L$ of  $L$ onto the cotangent bundle $T^* \mathsf{M}^n$ is the vector field defined as $\Omega^{-1} \theta_L$. In canonical coordinates, this yields formula
\begin{equation}\label{vert}
 \widetilde L = p_q L^q_1 \pd{}{p_1} + \dots + p_q L^q_n \pd{}{p_n},   
\end{equation}
 where $L^i_j$ are the components of $L$ in coordinates $u^1, \dots, u^n$.  Notice that  $\widetilde L = 0$ if and only if $L = 0$.

The next lemma summarises some basic properties of vertical lifts.
\begin{Lemma}\label{lem2.3}
For arbitrary operator fields $L$, $M$ and vector field $\xi$ we have
    \begin{equation}\label{4.1:formula1}
      \widehat L \widetilde M = \widetilde{ML}.  
    \end{equation}
    \begin{equation}\label{4.1:formula2}
        [\widehat \xi, \widetilde L] = \widetilde{\mathcal L_\xi L}
    \end{equation}
    \begin{equation}\label{4.1:formula3}
        \widehat L \widehat \xi = \widehat{L\xi} + \widetilde{\mathcal L_\xi L}
    \end{equation}
    \begin{equation}\label{4.1:formula4}
        [\widetilde L, \widetilde M] = \widetilde{LM - ML}.
    \end{equation}
\end{Lemma}
\begin{proof}
First, notice that for $\widehat L, \theta, \theta_M$ one has
$$
\widehat L^* \theta = \theta_L \quad \text{and} \quad \widehat L^* \theta_M = \theta_{LM}.
$$
Using the relation $L \Omega^{-1} = \Omega^{-1} \widehat L^*$, we obtain
$$
\widehat L \widetilde M = \widehat L \Omega^{-1} \theta_M = \Omega^{-1} \widehat L^* \theta_M = \Omega^{-1} \theta_{LM} = \widetilde{LM},
$$
which proves \eqref{4.1:formula1}.

Now consider formulas for $\widehat \xi$ and $\widetilde L$. Substituting them in the r.h.s. of \eqref{4.1:formula2} we get
$$
\begin{aligned}{}
[- p_q \pd{\xi^q}{u^i} \pd{}{p_i} + \xi^i \pd{}{u^i}, p_s L^s_j \pd{}{p_j}] = p_q \Big( \pd{L^p_j}{u^s} \xi^s - \pd{\xi^q}{u^s} L^s_j + L^q_s \pd{\xi^s}{u^j}\Big) \pd{}{p_j} = p_q (\mathcal L_\xi L)^q_j \pd{}{p_j}
\end{aligned}
$$
We get exactly $\widetilde{\mathcal L_\xi}$ in the r.h.s. and, thus,  \eqref{4.1:formula2} is proved. 

Next, by straightforward computation we get
$$
\begin{aligned}
\widehat L \widehat \xi & = \widehat L \Big( - p_q \pd{\xi^q}{u^i} \pd{}{p_i} + \xi^i \pd{}{u^i} \Big) = - L^s_i \pd{\xi^q}{u^s} p_q \pd{}{p_i} + p_q \pd{L^q_i}{u^s} \xi^s \pd{}{p_i} - p_q \pd{L^q_s}{u^i} \xi^s \pd{}{p_i} + L_s^i \xi^s \pd{}{u^i} = \\
& = \underbrace{p_q \Big(\pd{L^q_i}{u^s} \xi^s + L^q_s \pd{\xi^s}{u^i} -  L^s_i \pd{\xi^q}{u^s} p_q\Big) \pd{}{p_i]}}_{\text{this is $= \widetilde{\mathcal L_\xi L}$}} + \underbrace{\Big( - p_q \pd{L^q_s}{u^i} \xi^s - p_q L^q_s \pd{\xi^s}{u^i}\Big) \pd{}{p_i} + L_s^i \xi^s \pd{}{u^i}}_{\text{this is $ - p_q \pd{L^q_s \xi^s}{u^i} \pd{}{p_i} + L^i_s \xi^s \pd{}{u^i} = \widehat{L\xi}$}}
\end{aligned}
$$
Thus, formula \ref{4.1:formula3} is proved. 

Finally, 
$$
[\widetilde L, \widetilde M] = [p_s L^s_j \pd{}{p_j}, p_q M^q_i \pd{}{p_i}] = p_s (L^s_q M^q_i - M^s_q L^q_i) \pd{}{p_i} = \widetilde{ML - LM},
$$
which gives \eqref{4.1:formula4} and completes the proof of Lemma. \end{proof}

The first statement of Theorem \ref{t1} follows from Proposition 17 in \cite{yano}.

Assume that $L, M$ are symmetries of each other, that is $ML - LM = 0$ and $\langle L, M \rangle (\xi, \eta) + \langle L, M \rangle (\eta, \xi) = 0$ for all $\xi, \eta$. Using formula \eqref{ii:eq2.3}, we rewrite the latter condition as 
\begin{equation}\label{4.1:p1}
0 = \mathcal L_{M\xi} L - \mathcal L_{L\xi} M + L \mathcal L_\xi M - M \mathcal L_\xi L = - \langle L, M\rangle (\xi, \cdot) - \langle L, M\rangle (\cdot, \xi)
\end{equation}
for all vector fields $\xi$.

By direct computation we get
$$
\begin{aligned}
   & (\widehat L \widehat M - \widehat M \widehat L) \widehat \xi = \widehat L \widehat{M\xi} - \widehat M \widehat {L\xi} + \widehat L \widetilde{\mathcal L_\xi M} - \widehat M \widetilde{\mathcal L_\xi L} = \\
   = & \widehat{LM\xi} - \widehat{ML\xi} + \widetilde{\mathcal L_{M\xi} L} - \widetilde{\mathcal L_{L\xi} M} + \widetilde{\mathcal L_\xi M L} - \widetilde{\mathcal L_\xi L M} = \\
   = & \widetilde{\mathcal L_{M\xi} L} - \widetilde{\mathcal L_{L\xi} M} + \widetilde{L \mathcal L_\xi M} - \widetilde{M \mathcal L_\xi L}.
\end{aligned}
$$
Here we used formulas \eqref{4.1:formula1}, \eqref{4.1:formula3} from Lemma \ref{lem2.3}. The r.h.s. of this relation vanishes by \eqref{4.1:p1}. 
Since the lifts $\widehat \xi$ span $T_{(p,u)}(T^* \mathsf{M}^n)$ for $p\ne 0$, we conclude $\widehat L \widehat M - \widehat M \widehat L = 0$.

Now let us show that the symmetric part of $\langle \widehat L, \widehat M \rangle$  vanishes on $T^*\mathsf{M}^n$. Again, for an arbitrary vector field $\xi$ on $\mathsf{M}^n$ we have
$$
\begin{aligned}
     & \widehat L [\widehat \xi, \widehat M \widehat \xi] + \widehat M [\widehat L \widehat \xi, \widehat \xi] - [\widehat L \widehat \xi, \widehat M \widehat \xi] = \widehat L [\widehat \xi, \widehat{M\xi}] + \widehat L [\widehat \xi, \widetilde{\mathcal L_\xi M}] + \widehat M [\widehat {L \xi}, \widehat \xi] + \widehat M [\widetilde{\mathcal L_\xi L}, \widehat \xi] - [\widehat {L\xi}, \widehat{M\xi}] - \\
     - & [\widehat{L\xi}, \widetilde{\mathcal L_\xi M}] - [\widetilde{\mathcal L_\xi L}, \widehat{M\xi}] - [\widetilde{\mathcal L_\xi L}, \widetilde{\mathcal L_\xi M}] = \widehat L \widehat{[\xi, M\xi]} + \widehat M \widehat{[L\xi, \xi]} - \widehat{[L\xi, M\xi]} + \widehat L \widetilde{\mathcal L_\xi \mathcal L_\xi M} - \\
     - & \widehat M \widetilde{\mathcal L_\xi \mathcal L_\xi L} - \widetilde{\mathcal L_{L\xi} \mathcal L_\xi M} + \widetilde{\mathcal L_{M\xi} \mathcal L_\xi L} - \widetilde{\mathcal L_\xi L \mathcal L_\xi M} + \widetilde{\mathcal L_\xi M \mathcal L_\xi L} = \widetilde{\mathcal L_{[\xi, M\xi]} L} + \widetilde{\mathcal L_{[L\xi, \xi]} M} + \widetilde{\mathcal L_\xi \mathcal L_\xi M L} - \\
     - & \widetilde{\mathcal L_\xi \mathcal L_\xi L M} - \widetilde{\mathcal L_{L\xi} \mathcal L_\xi M} + \widetilde{\mathcal L_{M\xi} \mathcal L_\xi L} - \widetilde{\mathcal L_\xi L \mathcal L_\xi M} + \widetilde{\mathcal L_\xi M \mathcal L_\xi L}.
\end{aligned}
$$
The r.h.s. of the formula is the vertical lift of the operator field given by 
$$
\begin{aligned}
   & \mathcal L_{[\xi, M\xi]} L + \mathcal L_{[L\xi, \xi]} M + \mathcal L_\xi \mathcal L_\xi M L - \mathcal L_\xi \mathcal L_\xi L M -\mathcal L_{L\xi} \mathcal L_\xi M + \mathcal L_{M\xi} \mathcal L_\xi L - \mathcal L_\xi L \mathcal L_\xi M + \mathcal L_\xi M \mathcal L_\xi L = \\
   = & \mathcal L_\xi \mathcal L_{M\xi} L - \mathcal L_\xi \mathcal L_{L\xi} M + \mathcal L_\xi \mathcal L_\xi (ML) - \mathcal L_\xi (M \mathcal L_\xi L) - \mathcal L_\xi \mathcal L_\xi (LM) + \mathcal L_\xi (L \mathcal L_\xi M) = \\
   = & - \mathcal L_\xi \big(\langle L, M\rangle(\xi, \cdot) + \langle L, M\rangle(\cdot, \xi)\big).
\end{aligned}
$$
Here we again used \eqref{4.1:p1}  and formulas \eqref{4.1:formula1} -- \eqref{4.1:formula3} from Lemma \ref{lem2.3}. As $L, M$ are symmetries of each other, the expression in the brackets is zero and, therefore,  $\langle \widehat L, \widehat M\rangle (\widehat \xi, \widehat \xi) = 0$.  Since complete lifts $\widehat \xi$ span $T_{(p,u)}(T^* \mathsf{M}^n)$ for $p\ne 0$,  we conclude that the symmetric part of   $\langle \widehat L, \widehat M\rangle$  vanishes identically,  which implies the second statement of Theorem.

Now we proceed to the third. Assume that $L, M$ are strong symmetries of each other. This implies that they are symmetries of each other and $[[L, M]]_{FN} = 0$. Due to the first statement of Theorem \ref{t1}, we get that $[[\widehat L, \widehat M]]_{FN} = 0$ and $\widehat L, \widehat M$ are symmetries of each other due to the second statement. Thus, $\widehat L, \widehat M$ are strong symmetries of each other.

Finally, consider an arbitrary vector field $\xi$ on $\mathsf{M}^n$. By direct computation we get
$$
\begin{aligned}
    \widehat L \widehat M \widehat \xi - \widehat{LM} \widehat \xi = \widehat{LM \xi} + \widetilde{\mathcal L_{M\xi} L} -  \widetilde{\mathcal L_\xi M L} - \widehat{LM \xi} - \widetilde {\mathcal L_\xi (LM)} = \widetilde{\mathcal L_{M\xi} L} - \widetilde{M \mathcal L_\xi L}.
\end{aligned}
$$
Due to \eqref{ii:eq2}, the r.h.s. is the vertical lift of the operator field $- \langle L, M\rangle (\xi, \cdot)$. Similar to the previous, we get that $(\widehat L \widehat M  - \widehat{LM} ) \widehat \xi = 0$ implying $\widehat L \widehat M  = \widehat{LM}$, as  required.
Theorem is proved.
\end{proof}

\begin{proof}[Proof of Corollary \ref{cor2.1}]
 We will use the following statements.

\begin{Lemma}[Lemma 3.3 in \cite{nijapp2}]\label{lem2.4}
Let $L$ and $M$ be invertible Nijenhuis operators. Then $[[L, M]]_{FN} = 0$ if and only if $LM^{-1}$ is Nijenhuis.
\end{Lemma}

\begin{Lemma}[See Section 2 in \cite{magri2}]\label{lem2.5}
Consider a pair of sympletic forms $\Omega$ and $\bar \Omega$. The bivectors $\Omega^{-1}$ and $\bar \Omega^{-1}$ yield  compatible Poisson tensors if and only if the operator $\bar \Omega^{-1} \Omega$ is Nijenhuis.    
\end{Lemma}

Now assume that $L$ and $M$ are Nijenhuis operators on $\mathsf{M}^n$. Due to Theorem \ref{t1}, the lifts $\widehat L, \widehat M$ are Nijenhuis operators on $T^* M^n$. Moreover, if $[[L, M]]_{FN} = 0$, then $[[\widehat L, \widehat M]]_{FN} = 0$.

For invertible $L, M$, we have that $\det \widehat L \neq 0, \det \widehat M \neq 0$. For $\Omega_L$ and $\Omega_M$ we have
$$
\Omega_{L^{-1}} (\cdot, \widehat L \widehat M^{-1}\cdot) = \Omega (\cdot, \widehat L^{-1} \widehat L \widehat M^{-1} \cdot) = \Omega (\cdot, \widehat M^{-1} \cdot) = \Omega_{M^{-1}}(\cdot, \cdot).
$$
That is we get that the $2$-forms are related via $\widehat L \widehat M^{-1}$. Due to Lemma \ref{lem2.4}, the operator $\widehat L \widehat M^{-1}$ is Nijenhuis. Applying Lemma \ref{lem2.5}, we get that the corresponding Poisson brackets are compatible, as stated.
\end{proof}


\section{Lifts of gl-regular Nijenhuis operators and their symmetries}\label{sect3}

We say that a linear operator $L: V \to V$, $\dim V = n$, is $\gl$-{\it regular} if one of the following equivalent conditions holds:
\begin{enumerate}
    \item The operators $\operatorname{Id}, L, \dots, L^{n - 1}$ form a basis of the centraliser of $L$. In other words, every operator $M$ that commutes with $L$ can be uniquely written as $M= \sum_{k=0}^{n-1} a_k L^k$, $a_k\in\R$.
    \item There exists  $\xi\in V$ such that the vectors $\xi, L\xi, \dots, L^{n - 1} \xi$ are linearly independent.  Such a vector $\xi\in V$ is called {\it cyclic}.
    \item There exists  $\alpha\in V^*$ such that the covectors $\alpha, L^*\alpha, \dots, (L^*)^{n - 1} \alpha$ are linearly independent. Such a covector $\alpha\in V^*$ is called {\it cyclic}.
    \item In the Jordan normal form of $L$ (perhaps, complex),  each eigenvalue of $L$ has a single Jordan block. 
\end{enumerate}
We say that an operator field $L$ on $\mathsf M$ is a $\gl$-{\it regular Nijenhuis} operator, 
if $L$ is $\gl$-regular at every point $\mathsf p\in\mathsf M$ and its Nijenhuis torsion vanishes. The diagonal Nijenhuis operator  $L=\textrm{diag}(u^1,\dots, u^n)$ is $\gl$-regular, if $u_i\ne u_j$.

A closed differential form $\alpha$ is called a {\it conservation law} for an operator field $L$, if $L^* \alpha$ is closed also. Since all  constructions in this paper are local, we assume that $\alpha = \ddd f$ so that $\ddd f$ is a conservation law if and only if  $\ddd (L^* \ddd f) = 0$. We say that a conservation law $\ddd f$  of a $\gl$-regular Nijenhuis operator $L$ is {\it regular},  if
$$
\ddd f, L^* \ddd f, \dots, (L^*)^{n - 1} \ddd f
$$
are linearly independent. In other words, $\ddd f$ is a cyclic covector for $L$.

The symmetries and conservation laws of $\gl$-regular Nijenhuis operators were studied in \cite{nij4}. Here we provide a  short list of the results we will need for further work:
\begin{enumerate}
    \item Every symmetry of a $\gl$-regular Nijenhuis operator is strong  by    \cite[Theorem 1.2]{nij4}. Thus, when we say {\it symmetry} in the context of $\gl$-regular Nijenhuis operators, we mean {\it strong symmetry}.
    \item Every conservation law of a $\gl$-regular Nijenhuis $L$ is, at the same time, a conservation law of every symmetry of $L$ by  \cite[Theorem 1.3]{nij4}.
    \item For a real-analytic $\gl$-regular Nijenhuis operator,  regular conservation law locally exists near any point.  In the smooth case, the existence of a regular conservation law is proved under some genericity assumptions (see comments after Remark 1.1 in \cite{nij4}).
    \item In coordinates $u^1, \dots, u^n$ such that $\ddd u^1$ is a regular conservation law and $\ddd u^i = (L^*)^{i - 1} \ddd u^1$,  by  \cite[Theorem 1.3]{nij4},  the Nijenhuis operator takes the form
    \begin{equation}\label{companion}
    L = \left(\begin{array}{ccccc}
     0 & \!\! 1 &  \!\! 0 & \dots & 0  \\
     0 &  \!\! 0 &  \!\! 1 & \ddots & \vdots  \\
     \vdots & \!\!  \vdots &  \!\! \ddots & \ddots & 0 \\
     0 &  \!\! 0 & \!\!  \dots & 0 & 1  \\
     \sigma_n & \!\!  \sigma_{n - 1} & \!\!   \dots  & \sigma_2 & \sigma_1  \\
    \end{array}\right).    
    \end{equation}
    We refer to \eqref{companion} as {\it second companion form} and $u^1, \dots, u^n$ as {\it second companion coordinates}.
    \item For a real-analytic $\gl$-regular Nijenhuis operator, each symmetry is parameterized by $n$ functions of a single variable. In the smooth case,  the same holds at almost every point, see \cite[Theorems 1.4 and 1.5]{nij4}.
\end{enumerate}

\begin{Remark}\label{rem3.1}
\rm{
If $L$ is a $\gl$-regular Nijenhuis operator, then its lift $\widehat L$ is not $\gl$-regular. Indeed, take second companion coordinates. Then, due to Remark \ref{rem2.2},  $\widehat L$ has exactly two Jordan blocks for each of its eigenvalues. 
}    
\end{Remark}

\begin{Ex}\label{ex3.1}
\rm{
If a Nijenhuis operator $L$ has simple real spectrum, then in an appropriate coordinate system, it can be brought to the diagonal form
$$
L = \left( \begin{array}{ccccc}
     f_1(u^1) & 0  & \dots & 0  \\
     0 & f_2(u^2)  & \dots & 0  \\
     \vdots & \vdots &  \ddots &\vdots  \\
     0 & 0  & \dots & f_n(u^n)
\end{array}\right).
$$
We see that $\ddd (L^* \ddd u^i) = 0$, thus, due to Remark \ref{rem2.2}, we have
$$
\widehat L = \left( \begin{array}{cc}
     L &  0_{n \times n}\\
     0_{n \times n} & L 
\end{array}\right).
$$

Assume now that the derivatives of $f_i$ are all different from zero.  The coordinate change $u^i_{\mathrm{new}} = f_i(u^i)$
brings $L$ to a simpler form $L=\operatorname{diag}(u^1,\dots,u^n)$. \weg{Then the  splitting theorem (Theorem 1.1 in \cite{nij4}) reduces the study of strong symmetry of $\widehat L$ to the study of strong symmetries 
of two-dimensional Nijenhuis operators of the form:
$$
\left( \begin{array}{cc}
     u & 0  \\
     0 & u 
\end{array}\right) = u \cdot\Id.
$$
in the space with coordinates $(u,p)$. 
}
A straightforward computation shows that 
strong symmetries of $\widehat L$ are
$
Q=\begin{pmatrix} R & 0 \\ 0 & R \end{pmatrix}$, where $R = \operatorname{diag}\Bigl( 
r_1(u^1, p_1), \dots, r_n(u^n, p_n)\Bigr)$.   In particular, this implies that every strong symmetry $Q$ of $\widehat L$ is a linear combination of $\widehat{\Id}, \widehat L, \dots, \widehat L^{n-1}$.
}    
\end{Ex}

\begin{Ex}\label{ex3.2}
\rm{
Let $
L = \left( \begin{array}{cc}
     0 & 1 \\
     0 & 0 
\end{array}\right)
$ in some coordinates $u^1, u^2$.
The lift $\widehat L$, in coordinates  $ p_1, p_2, u^1, u^2$, takes the form
$$
\widehat L = \left(\begin{array}{cccc}
     0 & 0 & 0 & 0 \\
     1 & 0 & 0 & 0 \\
     0 & 0 & 0 & 1 \\
     0 & 0 & 0 & 0 \\
\end{array}\right).
$$
Condition $R \widehat L - \widehat L R$ yields
$$
R = \left(\begin{array}{cccc}
     r_{11} & 0 & 0 & r_{21} \\
     r_{12} & r_{11} & r_{21} & r_{22} \\
     r_{32} & r_{31} & r_{41} & r_{42} \\
     r_{31} & 0 & 0 & r_{41} \\
\end{array}\right)
$$
The differential condition $ \langle \widehat L, R \rangle=0$ amounts then to the following relations
$$
\widehat L \,\pd{R}{p_2} = 0, \quad \widehat L \,\pd{R}{u^1}=0, \quad \pd{R}{p_2} =  \widehat L\,\pd{R}{p_1}, \quad
\pd{R}{u^1} =  \widehat L\,\pd{R}{u^2}.
$$
leading to the following general form of strong symmetries of $\widehat L$:
$$
R = \left(\begin{array}{cccc}
     r_{11}  & 0 & 0 & r_{21}  \\
     \pd{r_{11}}{p_1} p_2 + \pd{r_{11}}{u_2} u_1 + m_{11} & r_{11} & r_{21} & \pd{r_{21}}{p_1} p_2 + \pd{r_{21}}{u_2} u_1 + m_{22}\\
     \pd{r_{31}}{p_1}p_2 + \pd{r_{31}}{u_2} u_1 + m_{32} & r_{31}  & r_{41}  & \pd{r_{41}}{p_1} p_2 + \pd{r_{41}}{u_2} u_1 + m_{42}\\
     r_{31} & 0 & 0 & r_{41}  \\
\end{array}\right),
$$
where $r_{ij}, m_{ij}$ are arbitrary functions of $p_1$ and $u_2$.   In contrast to the previous Example \ref{ex3.1}, $R$ cannot be written as a polynomial in $\widehat L$ with functional coefficients.

}    
\end{Ex}



\section{Completely integrable systems associated with gl-regular Nijenhuis operators}\label{sect4}

We start with a general theorem which provides a method for constructing commuting functions on Poisson--Nijenhuis manifolds based on the concept of a symmetry (strong symmetry) of an operator field.

Let $N$ be a Nijenhuis operator on a Poisson manifold $(\mathsf M, \mathcal P)$.
 Assume that the operators $\Id, N, \dots, N^{n-1}$ are linearly independent at each point, and $n$ is the maximal number with this property\footnote{Equivalently, $n$ is the degree of the minimal polynomial of $N$ so that for every $k\ge n$ we can write $N^k$ as a linear combination of $\Id, N, \dots, N^{n-1}$ with functional coefficients.}.  Consider the space $\mathcal S(N)$ of strong symmetries of $N$  that can be written as linear combinations   
$
Q = h_1 \Id + h_2 N + \dots + h_n N^{n-1}
$ with functional coefficients $h_i\in  C^\infty(\mathsf M)$:  
$$
\mathcal S(N) =\left\{ Q = \sum_k h_k N^{k-1}~|~ \langle Q, N\rangle = 0\right\}.
$$  
 
\begin{Theorem}\label{t2}
Let $Q_1,\dots, Q_n\in \mathcal S(N)$ be a basis of $\mathcal S(N)$ in the sense that every $Q\in \mathcal S(N)$ can be written as a linear combination 
$Q=\sum_i f_i Q_i$, $f_i\in C^\infty (\mathsf M)$. If $N$ is self-adjoint w.r.t. $\mathcal P$, that is, 
 $\mathcal P(\alpha,N^* \beta) = \mathcal P(N^*\alpha, \beta)$  for all $\alpha, \beta \in T^*_x\mathsf M$, $x\in \mathsf M$, then
 for any  $Q\in \mathcal S(N)$, the coefficients $f_i$ Poisson commute.
\end{Theorem}

\begin{proof}
The condition that $Q=f_1Q_1 + \dots + f_n Q_n$ is a strong symmetry of $N$ is equivalent to
$$
0 = \langle N,  f_1Q_1 + \dots + f_n Q_n \rangle =
\sum_i f_i \langle N , Q_i\rangle + N^* \ddd f_i \otimes Q_i  - \ddd f_i \otimes NQ_i.
$$
Since $\langle N , Q_i\rangle=0$, we conclude that necessary and sufficient condition for $Q=\sum f_i Q_i$ to be a strong symmetry of $N$  can be written in the form 
\begin{equation}
\label{eq:nsc}
N^* \ddd f_i \otimes Q_i  = \ddd f_i \otimes NQ_i 
\end{equation}
Notice that $NQ_i\in \mathcal S(N)$ and, therefore, we can set $NQ_i = \sum_j c_i^j Q_j$    to get
$$
\sum_j \left( N^*\ddd f_j - \sum_i c_i^j \ddd f_i\right) \otimes Q_j =0. 
$$ 
Since $Q_1,\dots, Q_n$ are linearly independent, we conclude that
\begin{equation}
\label{eq:QN}
N^*\ddd f_j = \sum_{i} c_i^j \ddd f_i.
\end{equation}   
In particular,  $\operatorname{Span} (\ddd f_1,\dots, \ddd f_n)$ is an $N^*$-invariant subspace.

\begin{Lemma}\label{lem4.1}
The subspace $\operatorname{Span} (\ddd f_1,\dots, \ddd f_n)$ is $N^*$-cyclic in the sense that  there exists  $\alpha = \sum_j \alpha_j \ddd f_j$ such that $\operatorname{Span} (\ddd f_1,\dots, \ddd f_n) =
\operatorname{Span} \bigl(\alpha, N^*\alpha, \dots, (N^*)^{n-1}\alpha\bigr)$.
\end{Lemma}

\begin{proof}
Take $\alpha = \sum_j \alpha_j \ddd f_j \in \operatorname{Span} (\ddd f_1,\dots, \ddd f_n)$. Then
$$
N^* \alpha = N^* \left(\sum_j \alpha_j \ddd f_j\right) = \sum_{i} \widetilde \alpha_i \ddd f_i = \widetilde \alpha, \quad \mbox{where $\widetilde \alpha_i =
\sum_{j} \alpha_j c_i^j $}
$$
Equivalently, in matrix form we have $\alpha \overset{N^*}{\mapsto} \widetilde \alpha = \alpha \, C$,  where $\alpha$ and $\widetilde \alpha$ are understood as row-vectors with $n$ components and  $C =  \left( c_i^j \right)$. Thus, it is sufficient to show that the matrix $C$ is cyclic.  

Recall that the entries of $C$ are defined from the relations $NQ_j = \sum_j c_i^j Q_j$, which means that 
$C$ is the matrix of the linear operator $\mu_N:\mathcal Z(N)  \to  \mathcal Z(N)$, where $\mathcal Z(N) = \operatorname{Span}\left(\Id, N, \dots, N^{n-1}\right)$ and $\mu_N(Q) = NQ$.  It remains to notice that $\mu_N$ is obviously cyclic with $\Id$ being one of its cyclic vectors.
\end{proof}

Finally, we use another well-known fact from Linear Algebra. 

\begin{Lemma}\label{lem4.2}
Let $\mathcal P$ be a skew-symmetric bilinear form on a vector space $V$ and $A:V\to V$ be a linear operator such that 
\begin{equation}
\label{eq:selfadjoint}
\mathcal P(A\xi , \eta) = \mathcal P(\xi, A\eta), \quad \mbox{for all $\xi,\eta\in V$}.
\end{equation} 
Then every $A$-cyclic subspace of $V$ is isotropic w.r.t. $\mathcal P$. 
\end{Lemma}

\begin{proof}  Let $W = \operatorname{Span}\left(A^{k}\xi, \ k=0,1,2,\dots \right) \subset V$ be an $A$-cyclic subspace.
We need to check that $\mathcal P(A^s\xi, A^k \xi)=0$ for all $k$ and $s$.  First of all, notice that  $\mathcal P_A (\xi, \eta)\overset{\mathrm{def}} = \mathcal P(A\xi, \eta)$ is skew-symmetric. Indeed, 
$$
\mathcal P_A(\xi, \eta) = \mathcal P(A\xi, \eta) = \mathcal P(\xi, A\eta) = 
- \mathcal P(A\eta, \xi) = - \mathcal P_A(\eta, \xi).
$$
Without loss of generality we assume that $s \ge k$. Then using \eqref{eq:selfadjoint}, we get
$$
\mathcal P(A^s\xi, A^k \xi) = \mathcal P (A^{s - 1}\xi, A^{k + 1} \xi) = \dots = 
\begin{cases}
\mathcal P(A^p\xi, A^p \xi) = 0,  & \mbox{if $k\!+\!s = 2p$},\\
\mathcal P(A^{p+1}\xi, A^p \xi) = \mathcal P_A (A^p\xi, A^p \xi) = 0, & \mbox{if $k\!+\!s = 2p\!+\!1$,}
\end{cases}
$$
as required.  \end{proof}

Lemmas \ref{lem4.1} and \ref{lem4.2}  (for $A=N^*$)  imply that  the subspace $\operatorname{Span} (\ddd f_1,\dots, \ddd f_n)$ is isotropic w.r.t. the Poisson structure $\mathcal P$, that is, $f_1,\dots, f_n$ Poisson commute, as required.
\end{proof}

We now return to the settings of Sections \ref{sect2} and \ref{sect3} and apply this theorem in the situation where $\mathcal P$ is the canonical Poisson structure $\Omega^{-1}$ on the cotangent bundle $T^*\mathsf M$ and $N$ is the complete lift $\widehat L$ of a $\gl$-regular Nijenhuis operator $L$ defined on $\mathsf M$.  In particular, $n = \dim M$, since $L$ is $\gl$-regular. 

\begin{Remark}\label{rem:4.1}{\rm
Notice that in this case the pair $(\Omega^{-1}, \widehat L)$ defines a Poisson--Nijenhuis structure on $T^*\mathsf M$ in the sense of Magri {\it et al} \cite{MagriSchw}, that is,  the bilinear form $\Omega^{-1} \bigl(\cdot,  {\widehat L}^* \cdot \bigr)$ is not only skew-symmetric (as the form $\mathcal P(\cdot, N^*\cdot)$ from Theorem \ref{t2}) but also satisfies the Jacobi identity and, as a result,  defines a Poisson structure compatible with $\Omega^{-1}$.
}\end{Remark}

Our next goal is to explain that the space 
$$
\mathcal S(\widehat L) = \left\{ Q =  h_1 \widehat \Id + h_2 \widehat L + \dots + h_n \widehat L^{n-1} ~|~ \langle Q, \widehat L\rangle = 0\right\}   
$$ 
is, in fact, very large. This  allows us to construct many examples of integrable Hamiltonian systems on $T^* \mathsf M$, both classical and new.

\begin{Ex}\label{ex:4.1}
{\rm If  $L=\operatorname{diag}(u^1,\dots, u^n)$, then as explained  in Example \ref{ex3.1},  every strong symmetry $Q$ of $\widehat L$ is a linear combination of \  $\widehat{\Id}, \widehat L, \dots , \widehat L^{n-1}$.
More specifically, in canonical coordinates  $(p, u)$,  we have $\widehat L = \begin{pmatrix} L & 0 \\ 0 & L\end{pmatrix}$, and the
operators $Q\in \mathcal S (\widehat L)$ have the form
$Q=\begin{pmatrix} R & 0 \\ 0 & R\end{pmatrix}$, where $R = \operatorname{diag} (r^1(p_1, u^1),\dots, r^n(p_n, u^n))$.

If we choose basis elememts $Q_i = \begin{pmatrix} R_i & 0 \\ 0 & R_i\end{pmatrix} \in \mathcal S(\widehat L)$, with $R_i = 
\operatorname{diag} (s^1_i(p_1, u^1),\dots, s^n_i(p_n, u^n))$,  then the relation $\sum f_i Q_i = Q$ from Theorem \ref{t2} amounts to the matrix equation
$$
\begin{pmatrix}
s^1_1 & \dots & s^1_n\\
\vdots & \ddots & \vdots \\
s^n_1 & \dots & s^n_n\\
\end{pmatrix} \begin{pmatrix} f_1 \\ \vdots \\ f_n  \end{pmatrix} = \begin{pmatrix}
r^1 \\ \vdots \\ r^n
\end{pmatrix}
$$ 
with $s^i_j = s^i_j(p_i, u^i)$ and $r^i = r^i(p_i, u^i)$ known as generalised  St\"ackel separation relations \cite[ \S  4.2.2]{Bbook}. If we assume that $s^i_j$ depends on $u^i$ only and $r^i = p_i^2 + V_i(u^i)$, then we obtain  classical St\"ackel separable systems.   
}\end{Ex}

Notice that in this example,  operators $Q\in \mathcal S(\widehat L)$ are parametrised by $n$ functions  $r_1,\dots, r_n$ of two variables.
Let us explain that this is a general fact.  

First of all, notice that $\mathcal S(\widehat L)$ is an infinite-dimensional real vector space that contains the lifts $\widehat M$ of all symmetries $M\in \Sym L$.  Indeed, if $M=\sum f_i L^{i-1}$ is a symmetry of $L$, then $\widehat M =\sum f_i \widehat L^{i-1}$.  This follows from the fact that every $\gl$-regular operator $L$ admits a coordinate system $u^1,\dots, u^n$  such that $\ddd u^i$ is a conservation law for $L$ and, therefore, for each  $M\in  \Sym L$ (see the list of properties of $\gl$-regular operators in Section \ref{sect3}). In the corresponding canonical coordinate systems $(p,u)$, the lift of symmetries $M$ is a purely algebraic operation that does not involve any derivatives (see Remark \ref{rem2.2}), and the above statement becomes obvious.  
Moreover, $\widehat M$ is a strong symmetry of $\widehat L$  by item 3 of Theorem \ref{t1}. Hence $\widehat M\in \mathcal S(\widehat L)$. 

The next observation is that for $Q_1, Q_2 \in  \mathcal S(\widehat L)$, we have $Q_1Q_2 \in  \mathcal S(\widehat L)$, i.e., $Q$ is a commutative associative algebra under pointwise multiplication.   

Next, we describe an important example of an operator  $P\in  \mathcal S(\widehat L)$, which is not a complete lift of any $M\in \Sym L$ and which will be then used to generate other operators $Q\in  \mathcal S(\widehat L)$ suitable for the construction from Theorem \ref{t2}. 

Let $\ddd f$ be a regular conservation law of $L$ and consider the second companion coordinates  $u^1,\dots, u^n$ generated by $\ddd f$, i.e.,  such that $\ddd  u^k = L_{k-1}^* \ddd f$.   Let $p_1,\dots, p_n$ denote the canonical momenta corresponding to the coordinates $u^i$ and introduce the operator $P$ on $T^*\mathsf M$  (i.e. $(1,1)$-tensor field on $T^*\mathsf M$) by setting
\begin{equation}
\label{eq:defP}
P =  p_1 \widehat\Id +  p_2 \widehat L + \dots +  p_n \widehat L^{n-1}.
\end{equation}
Notice that at each point $(p,u)\in T^*\mathsf M$,  the operator $P: T_{(p,u)} (T^*\mathsf M) \to T_{(p,u)} (T^*\mathsf M)$ can be equivalently defined by two conditions
  \begin{itemize}
  \item[\rm(i)] $P \in \operatorname{Span}\left(\Id, \widehat L, \widehat L^2 , \dots, \widehat L^{n-1}\right)$; 
  \item[\rm(ii)] $P^*\ddd f = \theta$, where $\theta$ is the canonical Liouville form\footnote{Notice that $\ddd f$ can be treated as a $1$-form not only on $\mathsf M$ but also on $T^*\mathsf M$, so that the expression $P^*\ddd f$ makes sense.}. 
\end{itemize}
Indeed, if $P = \sum f_i \left(\widehat L^{i-1}\right)$,  then $P^* \ddd f = \sum f_i \left(\widehat L^{i-1}\right)^*\ddd f = \sum f_i \ddd u^i$. Hence  $P^*\ddd f = \theta$ if and only if $f_i = p_i$.


\begin{Proposition}\label{prop4.1}   
$P\in \mathcal S(\widehat L)$.
\end{Proposition}
\begin{proof} To show that $P$ is a strong symmetry of $\widehat L$ we use formula \eqref{eq:nsc} as a necessary and sufficient condition (see proof of Theorem \ref{t2})  assuming that in our setting $N = \widehat L$ and $Q_i = \widehat L^{i-1}$. 
We use the second companion coordinate system $u^1,\dots, u^n$  constructed from $\ddd f $, that is,
$\ddd u^k = (L^{k-1})^* \ddd f$. By definition, 
$
P = p_1 \widehat \Id + p_2 \widehat L + \dots + p_n \widehat L^{n-1}, 
$
where $p_1,\dots,p_n$ are the momenta dual to $u^1,\dots, u^n$.  

Then formula \eqref{eq:nsc} becomes 
\begin{equation}
\label{eq:condition}
\widehat L^* \ddd p_1 \otimes \widehat \Id + \widehat L^* \ddd p_2 \otimes \widehat L + \dots + \widehat L^* \ddd p_n \otimes \widehat L^{n-1} = 
\ddd p_1 \otimes \widehat L + \dots +  \ddd p_n \otimes \widehat L^n.
\end{equation}

Since $L$ takes the second companion form \eqref{companion} and $\widehat L = \begin{pmatrix} L^\top & 0 \\ 0 & L \end{pmatrix}$, it is easy to see that
$\widehat L^* \ddd p_1 = \sigma_n \ddd p_n$  and $\widehat L^* \ddd p_k =  \ddd p_{k-1} + \sigma_{n-k+1} \ddd p_n$, $k=2,\dots,n$.
Hence, the left hand side of \eqref{eq:condition} becomes
$$
\sigma_n \ddd p_n \otimes \widehat \Id + (\ddd p_1 + \sigma_{n-1} \ddd p_n)\otimes \widehat  L + \dots + (\ddd p_{n-1} + \sigma_1\ddd p_n) \otimes \widehat L^{n-1}  =
$$
$$
\ddd p_1 \otimes \widehat L + \dots + \ddd p_{n-1}\otimes \widehat L^{n-1} + \ddd p_n \otimes \Bigl(  \sigma_n \widehat \Id + \sigma_{n-1} \widehat L + \dots + \sigma_1 \widehat L^n    \Bigr)  =
$$
$$
\ddd p_1 \otimes \widehat L + \dots + \ddd p_{n-1}\otimes \widehat L^{n-1} + \ddd p_n \otimes \widehat L^n, 
$$
as required. This completes the proof of Proposition \ref{prop4.1}. \end{proof}

Now let $U \in \Sym L$.  Then $\widehat U$ and $P$ both belong to $ \mathcal S(\widehat L)$ and therefore any polynomial $f(\widehat U, P)$ (with constant coefficients) belongs to
$ \mathcal S(\widehat L)$ also, i.e., $Q=f(\widehat U, P)$ is suitable for construction from Theorem \ref{t2}.  Moreover, this is still true for any matrix analytic function $f(\widehat U, P)$ (provided the spectra of $\widehat U$ and $P$ agree, in the natural sense, with the domain of $f$). 
More generally,  every operator of the form    
\begin{equation}
\label{eq:Q}
    Q = f_1 (\widehat U, P) \, \widehat \Id  +  f_2 (\widehat U, P) \, \widehat L + \dots +  f_n (\widehat U, P)\, \widehat L^{n-1},
\end{equation}
where $f_i(\cdot,\cdot)$ are arbitrary polynomials  or appropriate analytic functions,  belongs to $\mathcal S(\widehat L)$.  Thus, $\mathcal S(\widehat L)$ contains a family of operators \eqref{eq:Q} parametrised by $n$ functions of two variables. Examples below show that each of these functions is essential.  

Thus, we conclude that integrable systems constructed in Theorem \ref{t2} are parametrised by $n(n+1)$ functions of two variables, which are responsible for the choice of  $Q_1,\dots, Q_n$ and  $Q$.    

If we are interested in Hamiltonian systems with first integrals polynomial in momenta,  this can be easily achieved by an appropriate choice of $Q_1, \dots , Q_n$ and $Q$.   The next statement is a straightforward corollary of Theorem \ref{t2}.

\begin{Corollary}\label{cor:4.1}  
Let $M_1, \dots , M_n \in \Sym (L)$ be pointwise linearly independent symmetries of $L$ and  $\widehat M_1, \dots, \widehat M_n \in \mathcal S(\widehat L)$ their complete lifts considered as a basis of $\mathcal S(\widehat L)$.    
Set  $Q = \widehat U_0 + \widehat U_1 P +  \dots 
+ \widehat U_k P^k \in \mathcal S(\widehat L)$, where $U_i\in \Sym L$, and consider the coefficients  $h_1(p,u), \dots, h_n(p,u)$  of the expansion of $Q$ with respect to the basis  $\widehat M_1, \dots, \widehat M_n$:
\begin{equation}
\label{eq:MPU}
h_1 \widehat M_1 + \dots + h_n \widehat M_n = \widehat U_0 + \widehat U_1 P +  \dots 
+ \widehat U_k P^k.
\end{equation}
Then $h_1,\dots, h_n$ Poisson commute and are polynomial in $p_1,\dots, p_n$ of degree $k$.    If $Q= \widehat U_k P^k$,  then these polynomials are homogeneous.  
\end{Corollary}

In the examples below, as well as in the next two sections,  we assume that a basis of $\mathcal S(\widehat L)$ is formed by complete lifts $\widehat M_1, \dots, \widehat M_n$, $M_i\in\Sym L$.   As we shall see, such a choice allows us to describe commuting Hamiltonians $h_1,\dots, h_n$ in terms of the original operator $L$ and its symmetry algebra $\Sym L$, i.e., without referring to their lifts.   

\begin{Ex}\label{ex4.2}{\rm
In the simplest case, when $L=\operatorname{diag} (u^1, \dots, u^n)$,  we necessarily have 
$$
\begin{aligned}
Q=\begin{pmatrix} R & 0 \\ 0 & R\end{pmatrix} \in \mathcal S(\widehat L), \quad  
&R = \operatorname{diag}\bigl(r^1(p_1, u^1), \dots, r^n(p_n, u^n)\bigr) \\
\widehat M_i = \begin{pmatrix} M_i & 0 \\ 0 & M_i\end{pmatrix}\in \mathcal S(\widehat L), \quad 
&M_i=\operatorname{diag} \bigl(s^1_i (u^1),\dots, s^n_i(u^n)\bigr),
\end{aligned}
$$
so that the relation $Q=\sum f_i \widehat M_i$ from Theorem \ref{t2} becomes simply
$R = \sum f_i M_i$
and immediately leads, as explained in Example \ref{ex:4.1}, to St\"ackel systems.   We also notice that for $\ddd f = \ddd \tr L =\ddd u^1 + \dots + \ddd u^n$, the operator $P$ takes a very simple form
$$
P = \begin{pmatrix} \widecheck P & 0 \\ 0 & \widecheck P
\end{pmatrix},  \quad \mbox{where $\widecheck P=\operatorname{diag}(p_1,\dots, p_n)$}.   
$$ 
In particular,  for $Q = P^2$ the matrix relation  $Q = \sum h_i(p,u) \widehat M_i$ amounts to St\"ackel orthogonal separability relations
$$
\begin{pmatrix}
S^1_1(u^1) & \dots & S^1_n (u^1) \\
\vdots & \ddots & \vdots \\
S^n_1(u^n) & \dots & S^n_n (u^n) \\
\end{pmatrix} \begin{pmatrix}  h_1(p,u) \\ \vdots \\ h_n(p,u)  \end{pmatrix} =
\begin{pmatrix}   p_1^2 \\ \vdots \\ p_n^2 \end{pmatrix}
$$ 
for geodesic flows.    
}\end{Ex}

The next example is a new integrable system, which corresponds to $L$ being a nilpotent Jordan block.

\begin{Ex}\label{ex4.3}
\rm{
Fix coordinates $u_1, u_2, u_3$ (we use lower indices for coordinates specifically in this example) and consider the Nijenhuis operator
$$
L = \left( \begin{array}{ccc}
     0 & 1 & 0 \\
     0 & 0 & 1 \\
     0 & 0 & 0 \\
\end{array}\right).
$$
The symmetries  $M_i$ for $L$ take the form  (see \cite[Theorem 1.4]{nij4})
$$
M_i = \left( \begin{array}{ccc}
     s_{1i} & u_2 s_{1i}' + s_{2i} & u_1 s_{1i}' + \frac{1}{2} u_2^2 s_{1i}'' + u_2 s_{2i}' + s_{3i} \\
     0 &  s_{1i} & u_2 s_{1i}' + s_{2i} \\
     0 & 0 &  s_{1i}\\
\end{array}\right), \quad i = 1, 2, 3.
$$
Here functions $s_{ij}$ are arbitrary functions of $u_3$.  Choosing the regular conservation law to be $\ddd u^1$ leads to the operator $P$ of the form
$$
P=\begin{pmatrix} \widecheck P^\top & 0 \\ 0 & \widecheck P \end{pmatrix}, \quad\mbox{where} \ \widecheck P= 
\begin{pmatrix}
   p_1 & p_2 & p_3  \\
      0 & p_1 & p_2  \\
      0 & 0 & p_1  \\
\end{pmatrix}.
$$
To find the explicit form for the operator $Q$ given by \eqref{eq:Q},  we choose $U=\begin{pmatrix} u_3 & u_2 & u_1 \\ 0 & u_3 & u_2 \\ 0 & 0 & u_3\end{pmatrix}$ and take three arbitrary functions $g_i(\cdot, \cdot)$ of two variables. Then 
$
Q = g_1(\widehat U, P) \widehat \Id +  g_2(\widehat U, P) \widehat L +  g_3(\widehat U, P) \widehat L^2  
$
takes the form  $
Q = \left(\begin{array}{cc}
     V^{\top} & 0  \\
     0 &  V
\end{array}\right),
$
where
$$
\begingroup\makeatletter\def\f@size{9}
V = \left( \begin{array}{ccc}
     g_1 & \pd{g_1}{u_3} u_2 + \pd{g_1}{p_1} p_2 + g_2 &  \frac{\partial^2 g_1}{\partial u_3 \partial p_1} u_2 p_2 + \frac{1}{2} \frac{\partial^2 g_1}{\partial p_1^2} p_2^2 + \frac{1}{2} \frac{\partial^2 g_1}{\partial u_3^2}  u_2^2  + \pd{g_1}{p_1} p_3 + \pd{g_1}{u_3} u_1 + \pd{g_2}{p_1} p_2 + \pd{g_2}{u_3} u_2 + g_3\\
     0 & g_1 & \pd{g_1}{u_3} u_2 + \pd{g_1}{p_1} p_2 + g_2 \\
     0 & 0 & g_1 \\
\end{array}\right)
\endgroup
$$
and $g_i = g_i(u_3,p_1)$.
The matrix relation $\sum h_i \widehat M_i = Q$ or, equivalently, $\sum h_i M_i = V$  then amounts to the following linear equation system
$$
S \left( \begin{array}{c}
     h_1\\
     h_2\\
     h_3 \\
\end{array}\right) = \left( \begin{array}{c}
     \frac{\partial^2 g_1}{\partial u_3 \partial p_1} u_2 p_2 + \frac{1}{2} \frac{\partial^2 g_1}{\partial p_1^2} p_2^2 + \frac{1}{2} \frac{\partial^2 g_1}{\partial u_3^2}  u_2^2+ \pd{g_1}{p_1} p_3 + \pd{g_1}{u_3} u_1 + \pd{g_2}{p_1} p_2 + \pd{g_2}{u_3} u_2 + g_3\\
     \pd{g_1}{u_3} u_2 + \pd{g_1}{p_1} p_2 + g_2\\ 
     g_1 \\
\end{array}\right),
$$
where $S$ (analog of the St\"ackel matrix $\bigl( S^i_j\bigr)$ from Example \ref{ex4.2})  is the $3\times 3$ matrix whose $i$-th column coincides with the last column of the matrix $M_i$ and, therefore, depends on three arbitrary functions $s_{1i}(u_3)$, $s_{2i}(u_3)$ and $s_{3i}(u_3)$.   The functions $h_1(u,p)$, $h_2(u,p)$ and $h_3(u,p)$ so obtained Poisson commute and are independent if $g_1 \neq 0$. 
}    
\end{Ex}

As an important particular case of Corollary \ref{cor:4.1}, we finally present a method for producing natural Hamiltonian systems on $T^*\mathsf M$ with $n$ Poisson commuting first integrals \eqref{eq:formmetric} quadratic in momenta.  First of all notice that relation \eqref{eq:MPU}  from Corollary \ref{cor:4.1} can be written in terms of the original operators $M_i$ and $U_j$ as follows
$$
h_1 M_1 + \dots + h_n  M_n =  U_0 +  U_1 \widecheck P +  \cdots 
+  U_k \widecheck P^k,
$$
where $\widecheck P$ can be understood as the ``descent'' of $P$ from $T^*\mathsf M$ to $\mathsf M$ in the following sense. If $P = \sum f_i(p, u) \widehat L^{i-1}$, then $\widecheck P = f_i(p,u) L^{i-1}$. In particular, if we consider the second companion coordinates $u^1,\dots, u^n$ such that $\ddd u^i = (L^{i-1})^* \ddd f$, then $\widecheck P  = p_1\Id + p_2 L + \dots + p_n L^{n-1}$.  Notice that $\widecheck P $ cannot be treated as an operator on $\mathsf M$ since the functions $f_i(p,u)$ depend not only on $u$, but also on $p$.  In other words, $\widecheck P $ should be formally understood as a linear combination of operators on $M$ with coefficients being linear functions in momenta. We also notice that in arbitrary coordinate system $u^1,\dots, u^n$, at each point $u\in \mathsf M$, $\widecheck P $ can be equivalently defined as
$$
\widecheck P  = p_1 B_1 + \dots + B_n p_n,   
$$
where $B_i \in \operatorname{Span}(\Id, L, \dots, L^{n-1})$ is the unique operator such that $B_i^* \ddd f = \ddd u^i$.  This rule provides an easy method to  express $\widecheck P $ in any coordinate system.   Examples \ref{ex4.2} and \ref{ex4.3} give explicit formulas for $\widecheck P$ in two important and somehow opposite cases when $L$ is either diagonal or nilpotent. Similar formulas can be easily derived for $L$ of arbitrary algebraic type. After these preparations, we are ready to describe a construction of $n$ Poisson commuting functions \eqref{eq:formmetric} quadratic in momenta.

Let $L$ be a $\gl$-regular Nijenhuis operator,   $\Sym L$  its symmetry algebra, and $\widecheck P$  be the operator constructed above from a regular conservation law $\ddd f$. Notice that in a neighbourhood of a non-singular point\footnote{Recall that this condition means that locally the eigenvalues of $L$ have constant multiplicities},  an explicit description of symmetries and conservation laws of $L$ 
is given in \cite{nij4}. In particular, it is shown that the symmetries and conservation laws of $L$ are both parametrised by $n$ arbitrary functions of one variable.    

\begin{Theorem}\label{t2'}
Choose a basis $M_1, \dots, M_n$  in $\Sym L$ and an arbitrary symmetry $U\in \Sym L$.   Then the functions $h_\alpha(p,u)$, $\alpha=1,\dots,n$, defined from the relation 
\begin{equation}
\label{eq:P2N_new}
h_1 M_1 + \dots + h_n  M_n = \widecheck P^2 + U,  
\end{equation}
have the form  
$$
h_\alpha (p,u) = \tfrac{1}{2}  g_\alpha(p, p) + V_\alpha (u),  \quad     g_\alpha(p,p)= (g_\alpha)^{qs}(u)p_qp_s,
$$ 
and Poisson commute on $T^*\mathsf M^n$ with respect to the canonical Poisson structure. 
\end{Theorem}
      
This statement is a particular case of Corollary \ref{cor:4.1}.  Notice that $M_1,\dots, M_n$ and $U$ serve as parameters of this construction, and each of them depend on $n$ functions of one variables, so that similar to St\"ackel construction for integrable geodesic flows with potential (see formula \eqref{eq:St_intro} in Introduction), our construction is also parametrised by $n(n+1)$ functions of one variable.   Further properties of $h_\alpha$'s will be discussed in the next section in the context of the Hamilton--Jacobi integration method.

 \weg{     
Theorem \ref{t2} states that the functions $h_i$ from \eqref{eq:P2N} Poisson commute.  Example \ref{ex4.2}  shows that for $L=\operatorname{diag} (u^1,\dots, u^n)$ and $\ddd f = \ddd \tr L = \ddd u^1 +\dots + \ddd u^n$,  the corresponding Hamiltonian system will be orthogonally separable and of standard St\"ackel type (see Remark \ref{rem:5.1}).   The next theorem describes the properties of such systems corresponding to arbitrary $\gl$-regular operators $L$, not necessarily diagonalisable.  
}





\section{Hamilton--Jacobi for geodesic flows  with potentials}\label{sect5}

The concept of separation of variables is related to the Hamilton-Jacobi equation. In general, it provides one of the equivalent formulations of classical mechanics, which dates back to the XIX century \cite{nakane}. 

We will limit our discussion to a very specific class of systems. Namely, we consider {\it geodesic flows with potential}, i.e.,  Hamiltonian systems on $T^*\mathsf{M}^n$ with the canonical Poisson structure and Hamiltonians of the form
\begin{equation}\label{eq:hamH}
H(p, u) = \frac{1}{2} \sum_{q,s}g^{qs}(u) p_q p_s + V(u), \quad \operatorname{det} (g^{qs}) \neq 0,
\end{equation}
where $(p,u)=(p_1, \dots , p_n, u^1, \dots, u^n)$ are canonical coordinates on $T^*\mathsf M$.
Recall that $u^1, \dots, u^n$ are said to be {\it orthogonal separating} coordinates (or that $H$ admits  {\it orthogonal separation} in these coordinates), if 
\begin{enumerate}
    \item the matrix $g^{qs}(u)$ is diagonal; 
    \item the Hamilton--Jacobi equation
    \begin{equation}
    \label{eq:HJ}
    \frac{1}{2} g^{qs} \pd{W}{u^q} \pd{W}{u^s} + V(u)= c_1
    \end{equation}
    has a solution $W(u,c)=W(u^1,\dots, u^n, c_1,\dots, c_n)$ of the form
    $$
    W(u, c) = W_1(u^1, c) + \dots + W_n(u^n, c), \quad \det \Bigg( \frac{\partial^2 W}{\partial u^i \partial c_\alpha}\Bigg) \neq 0.
    $$
    that is, each function $W_i$ depends on $n + 1$ variables, $n$ parameters $c_1,\dots,c_n$ and $u^i$.
\end{enumerate}

For our purposes in this section, recall some basic properties of the square root of a linear operator $M: \R^n \to \R^n$. The square root of $M$ is an operator $R:\R^n \to \R^n$, such that $R^2 = M$. Depending on the Jordan normal form of $M$, there can be no such roots, finite set or even continuous families of such matrices $R$. 

We say that $R$ is a {\it good} root of $M$ if there exists a polynomial  $p(t)$ (Sylvester-Lagrange polynomial, see Chapter 8, \S 6 in \cite{gantmaher}) with constant coefficients such that $R = p(M)$. For example, if the eigenvalues of $M$ are either complex of real and positive, then such root does exist.  Moreover,  in a neighbourhood  $\mathcal U(M)$ of $M$ in $\gl (\R^n)$ there exists a unique real analytic matrix function $f : \mathcal U (M)  \to \gl(n,\R)$ such that $\widetilde R = f (\widetilde M)$, where $\widetilde R$ is a good root of $\widetilde M$ and $f(M) = R$. Throughout the rest of the paper,  $\sqrt{M}$ denotes a good root of $M$.

Our goal in this section is to study Poisson commuting functions $h_\alpha$ constructed in Theorem \ref{t2'}.       
Example \ref{ex4.2}  shows that for $L=\operatorname{diag} (u^1,\dots, u^n)$ and $\ddd f = \ddd \tr L = \ddd u^1 +\dots + \ddd u^n$,  the corresponding Hamiltonian system will be orthogonally separable and of standard St\"ackel type.   The next theorem describes properties of such systems in the case when $L$ is an arbitrary $\gl$-regular Nijenhuis operator, not necessarily diagonalisable.

\begin{Theorem}\label{t3}  
Let 
$h_\alpha(p, u) = \frac{1}{2} \sum_{s,q} (g_\alpha)^{sq}(u) p_s p_q + V_{\alpha}(u)$,  $\alpha=1,\dots, n$, 
be the functions obtained from relation \eqref{eq:P2N_new}. Then  
\begin{enumerate}
 \item  These functions are functionally independent on $T^*\mathsf M$ and generically $\det (g_1)^{sq}\ne 0$ so that  we obtain an integrable geodesic flow\footnote{In fact, if it happens that $\det (g_1)^{sq}= 0$, then we can always change the basis $M_1,\dots, M_n$ using a suitable transition matrix with constant entries in such a way that the new $g_1$ will be non-degenerate and will correspond to a certain (pseudo-)Riemannian metric.  Equivalently,  if $g^{ij} = \sum a_k g_k^{ij}$ then $\det g^{ij} \ne 0$ for generic coefficients $a_i\in \R$, so that a generic linear combination of the commuting Hamiltonians $h_i$ always correspond to a geodesic flow with potential.} with a potential.
    \item  $L$ is self-adjoint with respect to all contravariant symmetric forms $g_\alpha$, that is,
    $$
    L^j_p (g_\alpha)^{pk} = (g_\alpha)^{jp} L_p^k.
    $$
    \item Assume that the good root
    $R(c,u)=\sqrt{ c_1 M_1(u) + \dots + c_n M_n(u) - U(u)}$ exists, perhaps locally, as a smooth function of $u=(u^1,\dots,u^n)$ and real parameters $c=(c_1,\dots,c_n)$, and $\det R(c,u)\ne 0$.
    Then   
    the differential form 
    $
     R^*\, \ddd f
    $
    is closed for all values of parameters $c_1,\dots,c_n$, that is, there exists a function $W(u^1,\dots, u^n; c_1,\dots, c_n)$ such that 
    \begin{equation}
    \label{eq:dW}
    \ddd W = \left(\sqrt{ c_1 M_1 + \dots + c_n M_n - U} \right)^*\, \ddd f,
    \end{equation}
where the differential $\dd =\dd_u$ is taken only in coordinates $u^i$ and $c_i$ are treated as parameters. This function $W(x,c)$ is a solution of the Hamilton-Jacobi equation, moreover, $\det \Bigl( \frac{\partial^2 W}{\partial u^i \partial c_\alpha}\Bigr) \neq 0$.
\end{enumerate}
\end{Theorem}

\begin{proof}[Proof of Theorem \ref{t3}]
As before, we will use a special canonical coordinate system $(p, u)$ such that 
$\ddd u^k = (L^*)^{k-1} \ddd f$.     In this coordinate system, $L$ takes the second companion form, and 
the commuting Hamiltonians $h_i(p,u)$ are defined from the matrix identity
\begin{equation}
\label{eq:mainid}
\sum_\alpha h_\alpha  M_\alpha = \left(p_1\Id + p_2  L + \dots + p_n L^{n-1}\right)^2 +  U.
\end{equation}
This immediately implies that $h_\alpha = \frac{1}{2} g_\alpha(p,p) + V_\alpha(u)$, where $g_\alpha(p,p) = \sum (g_\alpha)^{sq} p_s p_q$ are some quadratic forms in momenta.

The statements 1 and 2 concern, in fact, linear-algebraic properties of the collection $g_1, \dots, g_n$  of quadratic forms obtained from the relation 
$$
\sum \tfrac{1}{2} g_\alpha(p,p) \, M_\alpha = \left(p_1\Id + p_2  L + \dots + p_n L^{n-1}\right)^2,
$$
and the dependence on the $u$-variables can be ignored.  We also notice that under changing the basis $M_1,\dots, M_n$,  
the collection of $h_\alpha$ will undergo a certain invertible linear transformation $h_\alpha \mapsto h_\alpha^{\mathrm{new}} = \sum c_\alpha^\beta h_\beta$, which does not affect the properties we need to verify.  Therefore, without loss of generality we may assume that  $M_1=L^{n-1} , M_2 =L^{n-2}, \dots, M_n =\Id$.  Then the above relation can be rewritten in the from 
$$
\tfrac{1}{2} g_1(p,p) L^{n - 1} + \dots + \tfrac{1}{2} g_n(p,p) \Id =
\Big(p_n L^{n - 1} + \dots + p_1 \Id\Big)^2,
$$
which immediately implies that the matrix $\Bigl( g_1^{sq} \Bigr)$ of $g_1$ is
$$
\Bigl( g^{sq}_1\Bigr) = \left( \begin{array}{ccccc}
     0 & \dots & 0 & 0 & 1  \\
     0 & \dots & 0 & 1 & * \\
     0 & \dots & 1 & * & *  \\
     & \iddots & & & \\
     1 & \dots & * & * & * \\
\end{array}\right),
$$
and, thus, $\det g_1^{sq}\ne 0$, as stated. The same will be true for a generic linear combination of the forms $g_\alpha$.

To prove that $h_\alpha(p,u)$ are functionally independent, it is sufficient to show the linear independence of the differentials 
of $h_i$ with respect to the $p$-variables only.  In other words, it is sufficient to verify that for a generic $p=(p_1,\dots,p_n)$ the vectors 
$
\xi_1 = \frac{\partial h_1}{\partial p} , \ \ 
\xi_2 = \frac{\partial h_2}{\partial p} , \ \dots \ , \   
\xi_n = \frac{\partial h_n}{\partial p} 
$
are linearly independent. 

In the proof of Theorem \ref{t2} we derived a necessary and sufficient condition \eqref{eq:QN} for a linear combination $\sum f_i Q_i$ to be a strong symmetry of the Nijenhuis operator $N$.  In our situation,  $Q_i = \widehat M_i$, $f_i=h_i$ and $N = \widehat L$,  and this formula can be written as  $\widehat L^* \ddd h_\beta  = \sum_\alpha c^\alpha_\beta \ddd h_\alpha$, where $c^\alpha_\beta$ are defined from the relations $L M_\beta  = \sum_j c^\alpha_\beta M_\alpha$.    
Taking into account the block-diagonal form of $\widehat L=\begin{pmatrix} L^\top & 0 \\ 0 & L  \end{pmatrix}$,  we get
\begin{equation}
\label{eq:b11}
L ^s_q \frac{\partial h_\beta}{\partial p_q} = c_\beta^1 \frac{\partial h_1}{\partial p_s} + c_\beta^2 \frac{\partial h_2}{\partial p_s} + \dots + c_\beta^n \frac{\partial h_n}{\partial p_s},  \quad \mbox{where $\frac{\partial h_\beta}{\partial p_s} = (\xi_\beta)^s = (g_\beta)^{sj} p_j$}
\end{equation}
or, shortly, 
$$
L\xi_\beta = c_\beta^1 \xi_1 + \dots + c_\beta^n \xi_n.
$$
This shows that  $\operatorname{Span}(\xi_1, \dots, \xi_n)$ is $L$-invariant.   It remains to use the fact that $g_1$ is non-degenerate so that $\xi_1$ can be an arbitrary tangent vector, if $p$ is appropriately chosen.  Thus,  if $\xi_1$ is $L$-cyclic, then the subspace $\operatorname{Span}(\xi_1, \dots, \xi_n)\subset T_u\mathsf M$ contains $\xi_1, L\xi_1, L^2\xi_1, \dots , L^{n-1}\xi_1$ and therefore coincides with the whole of the tangent space. Hence $\dim\operatorname{Span}(\xi_1, \dots, \xi_n) = n$ and $\xi_1, \dots, \xi_n$ are linearly independent, as required.   

Next, let us show that $L$ is self-adjoint w.r.t. each $g_\alpha$.  Indeed, \eqref{eq:b11}  holds identically for all $p=(p_1,\dots,p_n)$ and collecting the terms with $p_j$ we get  
\begin{equation}
\label{4.3:lava}
L^s_q (g_\alpha)^{qj} = c_\alpha^1 (g_1)^{s j} + \dots + c_\alpha^n (g_n)^{s j}.    
\end{equation}

For each $\alpha$, the r.h.s is symmetric in $s$ and $j$, which is equivalent to the third statement of  Theorem \ref{t3}.

Now let us move to the third statement. Fix the values of the integrals $h_\alpha$ by setting $h_\alpha = c_\alpha$. Then  \eqref{eq:mainid} becomes
\begin{equation}
\label{eq:to_resolve}
c_1 M_1 + \dots + c_n M_n - U = \left( p_1  \Id + p_2  L + \dots + p_n  L^{n-1} \right)^2.
\end{equation}

We need to locally resolve this relation to express $p_i$ as functions of $c=(c_1,\dots, c_n)$ and $u=(u^1,\dots, u^n)$. 
It easily follows from the implicit function theorem that locally (i.e., in a neighbourhood of every triple $(p,u,c)$ satisfying this relation) such a resolution  $p_i=p_i(u,c)$ always exist and is unique provided $\det (p_1  \Id + p_2  L + \dots + p_n  L^{n-1} )\ne 0$.

If $R=\sqrt{c_1 M_1 + \dots + c_n M_n - U}$ is a {\it good} root, then  by definition, $R$ is a polynomial of $c_1 M_1 + \dots + c_n M_n - U$ with some coefficients depending on $u$ and $c$.   But each $M_i$ and $U$ are polynomials in $L$ (with functional coefficients), hence $R$ itself  is a polynomial in $L$ and, therefore, we have
$$
R = s_1(u, c) \Id + s_2(u, c) L + \dots + s_n(u,c) L^{n-1}. 
$$
for some smooth functions $s_i=s_i(u,c)$.  Since resolving \eqref{eq:to_resolve} is locally unique, we conclude that $p_i=s_i(u,c)$ is a suitable resolution of  \eqref{eq:to_resolve}.

Now consider the $1$-form  
$$
\left(\sqrt{\sum c_\alpha M_\alpha  - U} \right) \ddd f = R^* \ddd f = 
s_1 \ddd f + s_2 L^*\ddd f + \dots + s_n( L^*)^{n-1} \ddd f = \sum_i s_i(u,c) \ddd u^i.
$$
The fact that this form is closed can be explained in two different ways. First notice that $\sum c_\alpha M_\alpha  - U$ is a symmetry of $L$, therefore $R$ is a symmetry also (since a real analytic function of a symmetry is a symmetry).  Then every conservation law of $L$ is, at the same time, a conservation law of $R$ and, therefore, $R^*\ddd f$ is closed.

Alternatively, we may use the fact that each level surface $\mathcal X_c = \{ h_\alpha(p,u)=c_\alpha, \ \alpha=1,\dots, n\}$ is Lagrangian.  If we resolve these equations w.r.t. $p$ and represent  $\mathcal X_c$ as a graph, that is, $\mathcal X_c =\{ p_i = p_i (u, c), \ i=1,\dots,n\}$, then the {\it Lagrangian property} is equivalent to the closedness of the $1$-form $\sum_i p_i(u,c) \ddd u^i$.  In our case by construction, $s_i(u,c)=p_i(u,c)$ and hence  $R^*\ddd f$ is closed.

Finally,  if $W(u,c)$ is such that $\ddd W(u,c) = S^*\ddd f = \sum_i s_i(u,c) \ddd u^i$, then  $\frac{\partial W}{\partial u^i} = s_i(u,c) = p_i$ and
$$
\frac{1}{2} (g_1)^{ij} \frac{\partial W}{\partial u^i} \frac{\partial W}{\partial u^j} + V_1(u)
= \frac{1}{2} (g_1)^{ij} p_i p_j + V_1 (u) = h_1(p,u) = c_1.
$$ 
Thus, $W(u,c)$ satisfies the Hamilton--Jacobi equation.  Also we notice that  $\frac{\partial^2 W}{\partial u^i \partial c_\alpha} = \frac{\partial s_i}{\partial c_\alpha}$ and by differentiating 
$\sqrt {\sum c_\alpha M_\alpha - U} = \sum_i s_i(u,c) L^{i-1}$ with respect to all $c_\alpha$ we come to the relations
$$
\frac{1}{2} M_\alpha R^{-1} = \sum_i \frac{\partial s_i}{\partial c_\alpha}  L^{i-1},
$$
which mean that $\Bigl( \frac{\partial s_i}{\partial c_\alpha}\Bigr)$ is a transition matrix between the bases 
$\frac{1}{2} M_1 R^{-1},\dots, \frac{1}{2} M_n R^{-1}$ and $\Id, L, \dots, L^{n-1}$. Hence, $\det \Bigl(\frac{\partial^2 W}{\partial u^i \partial c_\alpha} \Bigr) \ne 0$,  completing the proof of Theorem \ref{t3}.
\end{proof}


\begin{Ex}\label{ex5.1}
\rm{
In dimension two,  any Hamiltonian $H$ of form \eqref{eq:hamH}, which admits a quadratic first integral $F$,  can be brought, near a generic point, to one of the three normal forms (\cite{BMP2012}):
\begin{enumerate}
    \item Liouville case:
 $$
 \begin{aligned}
  H(u,p) & = \frac{p_1^2 - \varepsilon p_2^2}{a(u^1) - b(u^2)} + \frac{v_1(u^1) - \varepsilon v_2(u^2)}{a(u^1) - b(u^2)}, \\
  F(u,p) & = \frac{\varepsilon a(u^1) p_2^2 - b(u^2) p_1^2}{a(u^1) - b(u^2)} + \frac{\varepsilon a(u^1) v_2(u^2) - b(u^2) v_1(u^1)}{a(u^1) - b(u^2)},   
 \end{aligned}
 $$
 where $a, b, v_1, v_2$ are arbitrary smooth functions and $\varepsilon= \pm 1$.
    \item Complex Liouville case:
    $$
    \begin{aligned}
    H(u,p) & = \frac{2 p_1 p_2}{a(u^1, u^2)} + \frac{v_2(u^1, u^2)}{a(u^1, u^2)}, \\
    F(u,p) & = p_1^2 - p_2^2 - 2 \frac{b(u^1, u^2) p_1 p_2}{a(u^1, u^2)} - b(u^1, u^2)\frac{v_2(u^1, u^2)}{a(u^1, u^2)} - v_1(u^1, u^2),    
    \end{aligned}
    $$
    where $a(u^1, u^2) + \mathrm i \, b(u^1, u^2)$ and $v_1(u^1, u^2) +\mathrm i \, v_2(u^1, u^2)$ are holomorphic functions of the complex variable $z = u^1 + \mathrm i\, u^2$.
    \item Jordan block case:
    $$
    \begin{aligned}
     H(u,p) & = \frac{2p_1 p_2}{1 + u^1 b'(u^2)} + \frac{u^1 v_1'(u^2) + v_2(u^2)}{1 + u^1 b'(u^2)}, \\
     F(u,p) & = p_1^2 - \frac{2 b(u^2) p_1 p_2}{1 + u^1 b'(u^2)}  - b(u^2) \frac{u^1 v_1'(u^2) + v_2(u^2)}{1 + u^1 b'(u^2)} + v_1(u^2),   
    \end{aligned}
    $$
    where $b, v_1, v_2$ are smooth functions of one indicated varaible, and $b'$ and $v'_1$ denote the derivatives.
\end{enumerate}
These normal forms correspond to the following choice of parameters in our construction
\begin{enumerate}
    \item Taking $L, M_1, M_2$ and $U$ in the form
    $$
    L = \begin{pmatrix}  u^1 & 0 \\ 0 & u^2  \end{pmatrix},   \quad 
     M_1 = \begin{pmatrix}
    a(u^1) & 0  \\
    0 & \varepsilon b(u^2) 
    \end{pmatrix}, 
    \quad       
    M_2 =\begin{pmatrix}
    1 & 0  \\
    0 & \varepsilon 
    \end{pmatrix}, 
    \quad      U =\begin{pmatrix}
    v_1(u^1) & 0 \\
    0 & v_2(u^2) 
    \end{pmatrix},
    $$
    and the regular conservation law  $\ddd f = \ddd u^1 + \varepsilon \ddd u^2$, we obtain  the Liouville case.
    \item Taking $L, M_1, M_2, U$ in the form
    $$
    L = \begin{pmatrix}
     0 & 1  \\
     - 1 & 0 
     \end{pmatrix}, \quad 
     M_1 =\begin{pmatrix}
    \! -b(u^1, u^2) & -a(u^1, u^2)  \\
     \,  a(u^1, u^2) & -b(u^1, u^2) 
     \end{pmatrix}, 
     $$
     $$
     M_2 = \begin{pmatrix}
     - 1 & 0  \\
     0 & - 1 
     \end{pmatrix},   \quad 
   U = \begin{pmatrix}
    v_1(u^1,u^2),  & -v_2(u^1, u^2)  \\
     v_2(u^1, u^2)  & v_1(u^1, u^2) 
     \end{pmatrix} 
    $$
(such that  $a+\mathrm i\,  b $   and $v_1 + \mathrm i \, v_2$ are holomorphic with respect to $z=u^1 + \mathrm i \,  u^2$)  and  the   regular conservation law   $\ddd f = \ddd u^2$, we obtain   the complex Liouville case.
    \item  Taking $L, M_1, M_2, U$ in the form 
    $$
    L = \begin{pmatrix}
     0 & 1  \\
     0 & 0 
     \end{pmatrix}, \quad  
     M_1 = \begin{pmatrix}
     b(u^2) & b'(u^2) u^1 + 1  \\
     0 & b(u^2) 
     \end{pmatrix}, \quad 
     M_2 =  \begin{pmatrix}
     1 & 0  \\
     0 & 1 
     \end{pmatrix}, \quad  
     U =\begin{pmatrix}
     v_1(u^2) & v_1'(u^2) u^1 + v_2(u^2)  \\
     0 & v_1(u^2) 
     \end{pmatrix}.
    $$
    and the regular conservation law $\ddd f = \ddd u^1$, we obtain  the Jordan block case.
\end{enumerate}
Thus, all the normal forms are obtained by our construction. 
}    
\end{Ex}

\begin{Ex}\label{ex5.2}{\rm
By a clever choice of an operator $L$, symmetries $M_1,\dots, M_n$ and conservation law $\ddd f$, one can make one of the Poisson commuting functions $h_i$ to be the kinetic energy of the flat  space of signature (3,1). 
We take the $\gl$-regular  Nijenhuis operator $L$ given by (this operator comes from  \cite[Theorem 3]{nijapp2}):
\begin{equation} 
\label{eq:V1} 
L=\begin{pmatrix}
0 & 0 & x_{1} & 0 
\\
 1 & 0 & x_{2} & 0 
\\
 x_{2} & x_{1} & \!\! 1{+}2 x_{3} \!\! & x_{4} 
\\
 0 & 0 & x_{4} & 1 
\end{pmatrix}. \end{equation}
We choose the conservation law $\ddd f$  with $f = \tfrac{1}{2} \trace L= x_3+1$ and set
$$
M_1=L^3(L^3 -L^2)^{-1}, \quad M_2=L^2(L^3 -L^2)^{-1}, \quad M_3=L(L^3 -L^2)^{-1}, \quad M_4=(L^3 -L^2)^{-1}.
$$ 
Then the commuting functions $h_i$ obtained from the relation  
$\sum _i h_i  M_i =  \widecheck P ^2$  
are given by 
\begin{eqnarray*}
h_{4} & =&  
-p_{1}^{2} x_{1}^{2}+2 p_{1} p_{2} x_{1} x_{2}-p_{2}^{2} x_{2}^{2}-p_{2}^{2} x_{4}^{2}+2 p_{2} p_{4} x_{1} x_{4}-p_{4}^{2} x_{1}^{2}+2 p_{2}^{2} x_{3}-2 p_{2} p_{3} x_{1}+p_{2}^{2}
, \\ h_{3}  &= & 
p_{1}^{2} x_{1}^{2}-2 p_{1} p_{2} x_{1} x_{2}-2 p_{1} p_{2} x_{4}^{2}+2 p_{1} p_{4} x_{1} x_{4}+p_{2}^{2} x_{2}^{2}+2 p_{2} p_{4} x_{2} x_{4}-2 p_{4}^{2} x_{1} x_{2}\\ &+& 4 p_{1} p_{2} x_{3}-2 p_{1} p_{3} x_{1}-2 p_{2}^{2} x_{3}+2 p_{2} p_{3} x_{1}-2 p_{2} p_{3} x_{2}+2 p_{2} p_{1}-2 p_{2}^{2}
, \\ h_{2} &=& 
-4 p_{1} p_{2} x_{3}+2 p_{1} p_{3} x_{1}+2 p_{2} p_{3} x_{2}+2 p_{3} p_{4} x_{4}-2 p_{4}^{2} x_{3}-4 p_{2} p_{1}+p_{2}^{2}-p_{3}^{2}-p_{4}^{2}
, \\ h_{1} &= &2 p_{2} p_{1}+p_{3}^{2}+p_{4}^{2}. 
\end{eqnarray*}
We see that the function $h_1$ is (twice)  the Hamiltonian of the standard flat metric 
$$
g = \left(\begin{array}{cccc}
     0 & 1 & 0 & 0  \\
     1 & 0 & 0 & 0 \\
     0 & 0 & 1 & 0 \\
     0 & 0 & 0 & 1
\end{array}\right)
$$
of signature $(1, 3)$. 
Of course, Theorem \ref{t2'} also allows us to construct   ``potential energies''
$U_1,\dots, U_4$ such that  the functions $f_i = h_i+ U_i$ Poisson commute. It is sufficient to resolve the relations 
$\sum f_i  M_i = \widecheck P ^2 + U$ for a certain symmetry $U$ of $L$.  For example, the choice $U=L^4(L^3 - L^2)^{-1}$ gives us the functions 
$$
U_{4} = -x_{1}^{2}, \quad U_{3} = x_{1}^{2}-2 x_{1} x_{2}, \quad U_{2} = 
2 x_{1} x_{2}+x_{4}^{2}-2 x_{3}-1, \quad U_{1} = 2 x_{3}+2. 
$$

}
\end{Ex}



\section{Killing tensors, integrable quasilinear systems and reciprocal transformations}\label{sect6}

In the case of the basis $L^{n-1}, \dots , L, \Id$ in $\Sym L$,  and $U = 0$, Theorem \ref{t2'} (see also Theorem \ref{t3})  yields quadratic in momenta Hamiltonians $h_\alpha$ from the relation
\begin{equation}
\label{eq:37}
h_1  L^{n-1} + \dots + h_{n-1} L + h_n\Id = \widecheck P^2.
\end{equation}
These Hamiltonians play a crucial role in the theory of geodesically equivalent metrics \cite{nijappl5, nijapp2} and can be written in the form
\begin{equation}
\label{eq:38}
h_1(u,p) = \tfrac{1}{2} \, g(p,p), \quad h_\alpha(u,p)  = \tfrac{1}{2} \, g(A_\alpha^*p,p), \quad \alpha=2,\dots,n,
\end{equation}
where the operators $A_\alpha$ are understood as Killing $(1,1)$-tensors of the metric $g$, $A_1=\Id$.  For the above choice of the basis these operators can be found from the relation
$$
\det (\Id - \lambda L) (\Id - \lambda L)^{-1} = A_1 + \lambda A_2 + \dots + \lambda^{n - 1} A_n. 
$$
Here $A_1 = \Id$ and $A_i A_j = A_j A_i$.   Equivalently, they can be obtained from the following recursion formula. Consider $\chi_L(t) = \det(t \, \Id - L ) = t^{n} - \sigma_1 t^{n - 1} - \dots - \sigma_n$. The recursion relation is
\begin{equation}\label{rec}
A_1 = \Id, \quad A_i = L A_{i - 1} - \sigma_{i - 1}\Id, \quad i = 2, \dots, n.  
\end{equation}
Notice that the definition of $A_i$ does not depend on the choice of the conservation law $\ddd f$, but this choice will affect the metric $g$.

The operators $A_i$ possess remarkable properties in the context of evolutionary equations of hydrodynamic type, the main of which is the integrability of the quasilinear PDE system
\begin{equation}
\label{eq:canonA}
u_{t_i} = A_i u_x,    \quad u=\begin{pmatrix} u^1(t_1,\dots,t_n)\\ \vdots \\ u^n(t_1,\dots,t_n) \end{pmatrix}, \ \ x=t_1.
\end{equation}
If $L$ is diagonal, then this system belongs to the class of integrable weakly nonlinear systems studied by Ferapontov \cite{Ferapontov1991,Ferapontov1991b}.  For an arbitrary Nijenhuis operator $L$,  this system can be obtained as a particular case of the cohomological construction of integrable hierarchies of hydrodynamic type suggested by Lorenzoni and Magri \cite{ml}.  

Let us summarise other properties of $A_i$ which, in particular, justify the integrability of \eqref{eq:canonA},  for details see \cite{nij4}, Theorem 1.6.

\begin{enumerate}

\item[P1:] Operator fields $A_i$ are symmetries of one another, and their Haantjes torsions vanish.

\item[P2:]  For  any $M \in \Sym L$, consider its expansion w.r.t. the standard basis $L^{n-1}, \dots, L,\Id$: 
    $$
    M = g_1 L^{n - 1} + \dots + g_n \Id.    
    $$
    Then $A^*_i \ddd g_1 = \ddd g_i$. In particular, $\ddd g_1$ is a common conservation law of $A_i$'s.  Moreover, locally  every common conservation law of $A_i$'s can be obtained in this way for an appropriately chosen symmetry $M$.

\item[P3:]  For any conservation law $\ddd f_1$ of $L$, consider  $\ddd f_k = (L^*)^k \ddd f_1$, $k = 0, 1, \dots, n - 1$. Then the operator field 
    $$
    A = f_1 A_1 + \dots + f_n A_n    
    $$
is a common symmetry of $A_i$.  Moreover, locally every common symmetry of $A_i$'s can be obtained in this way for an appropriately chosen conservation law $\ddd f_1$.
\end{enumerate}

The above construction can be applied to an arbitrary basis $M_1,\dots, M_n$ of $\Sym L$.  Then by Theorem \ref{t2'},  the quadratic in momenta  functions $\bar h_\alpha(u,p): T^*\mathsf M \to \R$  obtained from the relation
\begin{equation}
\label{eq:hiMi}
\bar h_1  M_1 + \dots + \bar h_n  M_n = \widecheck P ^2,
\end{equation}
Poisson commute.  Similar to the case of the basis $L^{n-1},\dots,L,\Id$ (see \eqref{eq:37} and \eqref{eq:38}), these functions can be written as
\begin{equation}
\label{eq:hiKi}
\bar h_1(u,p) =  \tfrac{1}{2} \bar g(p,p), \quad h_\alpha(u,p) =  \tfrac{1}{2} \bar g(K_\alpha^*p,p), \quad  \alpha=2,\dots,n,
\end{equation}
where $K_\alpha$ are Killing $(1,1)$-tensors of the metric $\bar g$, $K_1 = \Id$.  Thus, given a basis $M_1,\dots, M_n$ of  $\Sym L$,  we produce a certain metric $\bar g$ together with $n$ independent Killing $(1,1)$-tensors $K_1,\dots, K_n$.  The next theorem describes these Killing tensors and shows that they possess the same properties as the operators $A_1,\dots, A_n$.    

\begin{Theorem}\label{t4}
Let $M_1, \dots, M_n$ be a basis of $\Sym L$  and 
$$
M_i = c^i_1 L^{n - 1} + \dots + c^i_n \Id.
$$ 
The functions $c^i_j$ form an invertible matrix (transition matrix from the standard basis to $M_1,\dots, M_n$) and we denote $ \mathcal C^{-1} = \widetilde {\mathcal C} = \Bigl( \widetilde c\,^i_j\Bigr)$.  Consider the operator fields {
$ R_j = \widetilde c\,^1_j A_1 + \dots + \widetilde c\,^n_j A_n$}, where $A_i$ are defined by \eqref{rec}. Then
\begin{enumerate}
    \item The Killing $(1,1)$-tensors $K_i$ related to the basis $M_1,\dots, M_n$ via \eqref{eq:hiMi}, \eqref{eq:hiKi} are given by the formula
    $$
    K_i = R_i R_1^{-1}, \quad i = 1, \dots, n.
    $$
    In particular, these Killing tensors depend only on the choice of a basis in $\Sym L$, but not of the conservation law $\ddd f$ used in the definition of $ \widecheck P$.
    \item $K_i$ are symmetries of one another, and their Haantjes torsions vanish;
    
    \item For  any $M \in \Sym L$, consider its expansion w.r.t. the basis $M_1,\dots, M_n$: 
   $$
    M = g_1 M_1 + \dots + g_n  M_n.    
   $$
    Then $K^*_i \ddd g_1 = \ddd g_i$. In particular, $\ddd g_1$ is a common conservation law of $K_i$'s.  Moreover, locally   every common conservation law of $K_i$'s can be obtained in this way for an appropriately chosen symmetry $M$.

\item For any conservation law $\ddd r_1$ of $L$, consider  $\ddd r_k = M_k^* \ddd r_1$, $k = 0, 1, \dots, n - 1$. Then the operator field 
    $$
    K = r_1 K_1 + \dots + r_n K_n    
    $$
    is a common symmetry of $K_i$'s.  Moreover, locally every common symmetry of $K_i$'s can be obtained in this way for an appropriately chosen conservation law $\ddd r_1$.
 \end{enumerate}
\end{Theorem}

\begin{Remark}\label{rem6.1}
\rm{
In general the following statement is true (follows from calculations in \cite{80}).  Let $A$ and $B$ be Killing $(1,1)$-tensors of a (contravariant) metric $g$ such that $AB=BA$.  Then the following two conditions are equivalent: 
\begin{enumerate}
    \item the quadratic integrals $f_A = \tfrac{1}{2} g(A^*p, p)$ and $f_B = \tfrac{1}{2} g(B^*p, p)$ associated with $A$ and $B$ Poisson commute;
    
   \item $A$ and $B$ are symmetries of each other.
\end{enumerate}
In our case we derive the same property for $K_i$ (item 2 of Theorem \ref{t4}) in a different way. 
}    
\end{Remark}

The proof of Theorem \ref{t4} will be based on the concept of a (generalised) reciprocal transformation in the context of quasilinear systems. We briefly recall this concept following \cite{Ferapontov1991c,Ferapontov1989d}.  Consider a collection of operator fields $B_i$, $i = 1, \dots, n$,  and the corresponding system of quasilinear equations:
\begin{equation}\label{sys}
   u_{t_i} = B_i u_x, \quad i = 1, \dots, n. 
\end{equation}
As before, we assume that $B_1 = \Id$ and $t_1 = x$. Suppose that $B_i$'s admit $n$ common conservation laws $\ddd g^1, \dots, \ddd g^n$ so that locally, there exist functions $c^i_j$ such that
$B_j^* \ddd g^i = \ddd c^i_j$, $1 \leq i, j \leq n$. 

We assume that the matrix $\mathcal C$ with entries $c^i_j$ is non-degenerate. Notice that, by construction, the functions $c^i_j$ are defined up to addition of arbitrary constants; thus, this condition is in no way restrictive. It holds even if $\ddd g^i \equiv 0$ for each $i=1,\dots,n$; in this case, the entries of $\mathcal C$ are just constants. Denote $\mathcal C^{-1} = \widetilde {\mathcal C}= \Bigl(\widetilde c\,^i_j\Bigr)$ and fix a collection of auxiliary operators $R_i = \sum_j \widetilde c\,_i^j B_j$. The {\it reciprocal transform}  of \eqref{sys} is a new quasilinear system 
\begin{equation}\label{sys2}
u_{\tau_i} = K_i u_x, \quad i = 1, \dots, n,
\end{equation}
where $K_i = R_i R^{-1}_1$. The main property of the reciprocal transformation is that it {\it maps} solutions to solutions. This procedure works as follows: 
\begin{enumerate}
    \item Let $u(t)=u(t_1,\dots,t_n)$ be an arbitrary solution of  \eqref{sys} and define $1$-forms
    $$
    \alpha^i = c^i_1\bigl(u(t)\bigr) \ddd t_1 + \dots c^i_n \bigl(u(t)\bigr) \ddd t_n.
    $$
    \item By construction,  the forms $\alpha^i$ are closed. Hence, locally there exists functions $\tau_i (t)$ such that $\ddd \tau_i = \alpha^i$, $i=1,\dots,n$. 
    \item The matrix $\mathcal C(t)= \Bigl( c^i_j\bigl(u(t)\bigr)  \Bigr)$ coincides with $\pd{\tau}{t}= \Bigl( \pd{\tau_i}{t_j}\Bigr)$. Thus, since $\mathcal C(t)$ is invertible, locally there exists a well-defined inverse map $t(\tau)$. Then the function $\widetilde u(\tau) = u\bigl(t(\tau)\bigr)$ is a solution of \eqref{sys2}.
\end{enumerate}
In other words, knowing a solution of the initial system, we can find a solution of its reciprocal transform \eqref{sys2} in quadratures. Note that the reciprocal transformation is invertible, that is by construction the forms
$$
\widetilde \alpha^i = {\widetilde c\,}^i_1\bigl( \widetilde u(\tau)\bigr) \ddd \tau_1 + \dots + {\widetilde c\,}^i_n(\tau) \ddd \tau_n
$$
are closed (simply because the elements of $\mathcal C^{-1}$ are $\pd{t}{\tau}$). 

\begin{proof}[Proof of Theorem \ref{t4}]

We will use notation $M^i$ instead of $M_i$ for the basis of $\mathrm{Sym}\, L$. Let  $M^i = c^i_s L^{n - s}$ so that  $\mathcal C =\Bigl( c^i_s\Bigr)$ is the transition matrix between the bases $L^{n-1}, L^{n-2}, \dots, L, \Id$ and $M^1,\dots, M^n$.  The collections of commuting Hamiltonians $h_1, \dots, h_n$ and $\bar h_1, \dots,  \bar h_n$ are related by 
$$
\sum_s h_s L^{s-1} = \sum_i \bar h_i M^i = \widecheck P ^2
$$
which implies $h_s = \sum_s \bar h_i c^i_s$ or, equivalently, $\bar h_i = \sum_s \widetilde c\,^s_i \, h_s$.  In the above notation,   $h_s = g(A^*_s p, p)$ and $R_i =  \sum_s \widetilde c\,^s_i  A_s$. Hence, $\bar h_i  = g(R_i^* p, p)$ and, in particular, $\bar h_1 = \bar g(p,p)=g(R_1^* p, p)$ so that finally, 
$$
\bar h_i  = \bar g \left(\left( R_i R_1^{-1}\right)^* p, p\right),
$$
that is, 
$$
K_i = R_i R_1^{-1}, 
$$
which completes the proof of item 1.

We now recall an important property of operators $A_i$ that for any symmetry $M = \sum_s g_s L^{n-s}$, the function $g_1$ is a density of a common conservation law of $A_i$. Moreover, $\ddd g_s = A_s^*\ddd g_1$ (see property P2).
This means that in our case $c_1^i$, $i=1,\dots,n$ are densities of conservation laws of $A_i$, and the above transition matrix $\mathcal C = \Bigl( c^i_s\Bigr)$ coincides with the matrix $\mathcal C$ in the definition of a reciprocal transform generated by the densities $g^1=c^1_1, \dots, g^n=c^n_1$.  In other words, the Killing $(1,1)$-tensors $K_1,\dots, K_n$ are obtained from $A_1,\dots , A_n$ by a reciprocal transform.   

Since $A_1,\dots, A_n$ are symmetries of one another, the same is true for $K_1,\dots, K_n$  (this is a well known general property of a reciprocal transform, which, of course, can be easily checked).  Next, $K_i$ is a Haantjes operator. Indeed,  $K_i$  algebraically commutes with $L$ and therefore can be written as a polynomial of $L$ with functional coefficients. Since $L$ is Nijenhuis,  by the well-known property (see \cite{ob}, for example), the Haantjes torsion of $K_i$ vanishes. Thus, the second statement is proved in full.

Since $K_1=\Id, K_2,\dots, K_n$ and $A_1=\Id, A_2, \dots, A_n$ are connected  by a reciprocal transform,  we can use this fact to establish a natural relationship between the conservation laws of these two collections of operators\footnote{Or, equivalently, between the conservation laws of the corresponding quasilinear PDE systems $u_{t_i}= K_i u_x$ and $u_{\tau_s} = A_s u_x$ for $i,s=1,\dots, n$.}.  Namely,  the following general property holds.  If  $\mathcal C = \Bigl(c^i_s\Bigr)$ is the matrix of the corresponding reciprocal transformation and $\ddd g$ is a common conservation law for $K_i$ with $ K_i^* \ddd g=\ddd g_i$, then the functions $f_s = \sum_i g_i c^i_s$ satisfy
\begin{equation}
\label{eq:aboutA}
A_s^*\ddd f_1 = \ddd f_s
\end{equation}
and, in particular, $\ddd f = \ddd f_1$ is a common conservation law for $A_1, \dots, A_s$. Notice that the functions $g=g_1,g_2,\dots, g_n$ are connected with $f_1=f, f_2, \dots , f_n$ by means of the inverse matrix $\mathcal C^{-1}=\widetilde{\mathcal C}$, which defines the inverse reciprocal transform, so that this correspondence is bijective. 

Now it remains to apply property P3 of the operators $A_i$'s, which states that  \eqref{eq:aboutA} is equivalent to the fact that $M= \sum f_s L^{n-s}$ is a symmetry of $L$. It remains to note that $\sum_s f_s L^{n-s} = \sum_{s, i}  g_i c^i_s L^{n-s} = \sum_i g_i M^i$, which completes the proof of the third statement of Theorem \ref{t4}.

Let us proceed to the fourth statement. We start with two general lemmas. Let $A_1, \dots, A_n$ be operator fields that satisfy the following properties:
\begin{enumerate}
    \item $A_i$ pointwise span a commutative associative algebra, that is, $A_i A_j = \sum_k a_{ij}^k A_k$;
    \item For all $\mathsf p\in\mathsf M$ and almost all $\xi \in T_{\mathsf p} \mathsf M$, the vectors $A_i \xi$ are linearly independent;
    \item $A_i$ are symmetries of one another.
\end{enumerate}

\begin{Lemma}\label{lem6.1}
Let $A_i$ satisfy the above conditions. Consider the operator field $A = g_1 A_1 + \dots + g_n A_n$ and define
$T_A = \ddd g_1 \otimes A_1 + \dots + \ddd g_n \otimes A_n$. Then $A$ is a common symmetry of $A_i$ if and only if
\begin{equation}\label{qq2}
T_A(A_i\xi, \eta) = T_A(\xi, A_i\eta)    
\end{equation}
for all $A_i$ and $\xi, \eta \in T_{\mathsf p} \mathsf M$.
\end{Lemma}

\begin{proof}
By direct computation, we have
\begin{equation}\label{qq1}
\begin{aligned}
& \langle A_i, A \rangle (\xi, \eta) = AA_i [\xi, \eta] + [A_i \xi, A \eta] - A_i [\xi, A \eta] - A [A_i\xi, \eta] = \\
& = \sum_{j = 1}^n\Big( g_j A_j A_i [\xi, \eta] + [A_i \xi, g_j A_j \eta] - A_i [\xi, g_j A_j \eta] - g_j A_j [A_i\xi, \eta] \Big) = \\
& = \sum_{j = 1}^n g_j \langle A_i, A_j\rangle (\xi, \eta) 
 + \sum_{j = 1}^n \Big( A_i^* \ddd g_j \otimes A_j - \ddd g_j \otimes A_i A_j \Big) (\xi, \eta) = \\
& = \sum_{j = 1}^n \Big(A^*_i \ddd g_j - \ddd g_1 a_{i1}^j - \dots - \ddd g_n a_{in}^j\Big) \otimes A_j (\xi, \eta) + 
 \sum_{j = 1}^n g_j \langle A_i, A_j\rangle (\xi, \eta).
\end{aligned}    
\end{equation}
First, assume that $A$ is a common symmetry. Substituting $\eta = \xi$, we get that $\langle A_i, A_j \rangle (\xi, \xi)$ vanishes. At the same time, all $A_i\xi$ are linearly independent for almost all $\xi$. By continuity, we obtain
$$
A^*_i \ddd g_j - a_{i1}^j \ddd g_1 - \dots - a_{in}^j \ddd g_n = 0, \quad j = 1, \dots, n.
$$
At the same time
$$
\begin{aligned}
T_A(A_i\xi, \eta) - T_A(\xi, A_i\eta) = \sum_{j = 1}^n \Big(A_i^* \ddd g_j \otimes A_j - a_{1i}^1 \ddd g_1 - \dots - a_{ni}^n \ddd g_n \Big) = 0.
\end{aligned}
$$
Thus, if $A$ is a common symmetry, then \eqref{qq2} holds. Now assume that \eqref{qq2} holds. Going back to formula \eqref{qq1}, we get
$$
\langle A_i, A \rangle (\xi, \eta) = \sum_{j = 1}^n g_j \langle A_i, A_j\rangle (\xi, \eta).
$$
Thus, we get that $A$ is a common symmetry of $A_i$, as stated.
\end{proof}

\begin{Lemma}\label{lem6.2}
Let $A_i$, $i=1,\dots,n$, satisfy the above conditions, and $\ddd f$ be a common conservation law of them. Then $\ddd f$ is a conservation law for every common symmetry $A$ of the form $A = g_1 A_1 + \dots + g_n A_n$.
\end{Lemma}
\begin{proof}
By our assumptions, there exist functions $f_i$ such that $\ddd f_i = A_i \ddd f$. By direct computation, for an arbitrary pair of vectors $\xi, \eta$, we get
$$
\begin{aligned}
& \ddd(A^*\ddd f) (\xi, \eta) = \ddd (g_1 \ddd f_1 + \dots + g_n \ddd f_n) (\xi, \eta) = \ddd g_1 \wedge \ddd f_1 (\xi, \eta) + \dots + \ddd g_n \wedge \ddd f_n (\xi, \eta) = \\
= & \ddd f(T_A(\xi, \eta)) - \ddd f(T_A(\eta, \xi)).
\end{aligned}
$$
Substituting $A_i\xi, A_j \xi$ instead of $\eta$ and $\xi$ and using Lemma \ref{lem6.1}, we get
$$
\ddd f(T_A(A_i \xi, A_j \xi)) - \ddd f(T_A(A_j\xi, A_i \xi)) = \ddd f(T_A(A_j A_i \xi, \xi)) - \ddd f(T_A(A_i A_j\xi, \xi)) = 0.
$$
Since $A_1\xi, \dots, A_n\xi$ form a basis, we conclude that $\ddd(A^*\ddd f) (\xi, \eta)=0$, as required.
\end{proof}

As before, let  $g^1, \dots, g^n$ be the densities of common conservation laws for $A_i$. Consider a~symmetry $A = f_1 A_1 + \dots + f_n A_n$ and take the extended reciprocal transformation \footnote{This transformation exists due to Lemma \ref{lem6.2}} with the new $n + 1 \times n + 1$ matrix $\mathcal C_{\mathrm{ex}}$ from
$$
\begin{aligned}
\ddd \tau_1 &=   c_1^1 \ddd t_1 + \dots +  c_n^1 \ddd t_n + c_{n+1}^1 \ddd t_{n+1} ,\\
& \dots \\
\ddd \tau_n &=   c_1^n \ddd t_1 + \dots +  c_n^n \ddd t_n + c_{n+1}^n \ddd t_{n+1} ,\\
\ddd \tau_{n+1} & = \phantom{c_1^n \ddd t_1 + \dots +  c_n^n \ddd t_n +  c_{n+1}^n} \ddd t_{n+1}.
\end{aligned}
$$ 
The new extended matrices $\mathcal C_{\mathrm{ex}}$ and $\widetilde {\mathcal C}_{\mathrm{ex}}=\mathcal C^{-1}_{\mathrm{ex}}$ will take the form
$$
\mathcal C_{\mathrm{ex}} = \begin{pmatrix}
\mathcal C & \begin{matrix}  c_{n+1}^1 \\ \vdots \\ c_{n+1}^n \end{matrix} \\
0 \ \dots \ 0 & 1
\end{pmatrix}, \quad 
\widetilde {\mathcal C}_{\mathrm{ex}} = \mathcal C^{-1}_{\mathrm{ex}} =  \begin{pmatrix}
\widetilde{\mathcal C} & -\widetilde {\mathcal C}\begin{pmatrix}  c_{n+1}^1 \\ \vdots \\ c_{n+1}^n \end{pmatrix} \\
0 \ \dots \ 0 & 1
\end{pmatrix}
$$

The matrices $K_i$ of the corresponding reciprocal transform will be defined by the standard formulas
$$
K_i = R_i R_1^{-1},  \quad  i = 1,\dots, n, n+1. 
$$
It is easy to check that the matrices $R_1, \dots, R_n$ remain exactly the same so that $K_1, \dots, K_n$ do not change. We are interested in the additional operator $K=K_{n+1}$ which will be a common symmetry of $K_1, \dots, K_n$  (corresponding to the symmetry $A$ of the original system). This operator has the form
$$
K_{n+1} = \left( \sum_{i=1}^n\widetilde c\,^i_{n+1} A_i + A\right) R_1^{-1}
$$
where
$$
\widetilde c\,^i_{n+1} = - \sum_{j=1}^n\widetilde c\,^i_j  c^j_{n+1}
$$
Finally, the symmetry $K$ obtained from the symmetry $A$ has the form
\begin{equation}\label{arj}
\begin{aligned}
K & = \left(A - \sum_{i,j=1}^n\widetilde c\,^i_j  c^j_{n+1} A_i\right) R_1^{-1} = \left(A - \sum_{j=1}^n \left(\sum_{i=1}^n \widetilde c\,^i_j A_i\right) c^j_{n+1} \right) R_1^{-1}  = \\
& = \left(A - \sum_{j=1}^n R_j c^j_{n+1} \right) R_1^{-1}  = AR^{-1}_1 - \sum_{j=1}^n c^j_{n+1} K_j.        
\end{aligned}
\end{equation}
Thus, any common symmetry $A$ for $A_i$ generates a common symmetry $K$ for $K_i$. Notice that $K_i$ satisfies the conditions of Lemmas \ref{lem6.1} and \ref{lem6.2}. Thus, applying the inverse reciprocal transformation, we find that all common symmetries $K$ are obtained from some common symmetry $A$ of $A_i$. In other words, we have a bijection of the common symmetries. 

At the same time $A_i = c_i^s R_s$, $A = \sum_j f_j A_j$, so 
$$
AR_1^{-1} = \sum_{j = 1}^n f_j c_j^s K_s. 
$$
Substituting this into \eqref{arj}, we get
$$
K = \Big(\sum_{j = 1}^n f_j c_j^1 - c^1_{n + 1}\Big) K_1 + \dots + \Big(\sum_{j = 1}^n f_j c_j^n - c^n_{n + 1}\Big) K_n = r_1 K_1 + \dots + r_n K_n.
$$
Denoting the equations in the brackets as $r_i$, we get
$$
\begin{aligned}
\ddd r_i & = c_1^i \ddd f_1 + \dots + c_n^i \ddd f_n + f_1 \ddd c_1^i + \dots + f_n \ddd c_n^i - \ddd c_{n + 1}^i  = \\
& = \Big(c_1^i L^{n - 1} + \dots + c_n^i \Id \Big)^* \ddd f + \Big(f_1 A_1 + \dots + f_n A_n\Big)^* \ddd g_1 - A^* \ddd g_1 = \\
& = M^{i*} \ddd f.
\end{aligned}
$$
This completes the proof of  item 4 of Theorem \ref{t4}.  \end{proof}

Since the composition of reciprocal transforms is a reciprocal transform,    quasilinear systems related to different bases of the symmetry algebra $\Sym L$  can be obtained from each other by a suitable reciprocal transformation. Namely, we have the following corollary of Theorem \ref{t4}.

\begin{Corollary}\label{c1}
Let  $M_1,\dots, M_n$  and $\widetilde M_1,\dots, \widetilde M_n$ be two bases of $\Sym L$, and $K_1=\Id, K_2, \dots, K_n$, and $\widetilde K_1=\Id, \widetilde K_2, \dots, \widetilde K_n$ be the corresponding collection of Killing $(1,1)$-tensors constructed from these bases via  \eqref{eq:hiMi}, \eqref{eq:hiKi} (see Theorem  \ref{t4}). Define $c^i_j$ as $\widetilde M_i = c^i_1 M_1 + \dots + c^i_n M_n$. Then the (integrable) quasilinear systems
$$
u_{t_i} = K_i u_x \quad (i=1,\dots, n)\quad\mbox{and}\quad u_{\tau_i} = \widetilde K_i u_x \quad (i=1,\dots, n)
$$
are related by the reciprocal transformation defined by the conservation laws $\ddd g_i = \ddd c^i_1$.   
\end{Corollary}

If $L$ is not only $\gl$-regular, but also algebraically generic, that is, the eigenvalues of $L$ have constant multiplicities, then it easily follows from \cite{nij4} that  in a suitable coordinate system, $\Sym L$ admits a basis $\widetilde M_1,\dots, \widetilde M_n$ that consists of constant matrices.  In this case, the Killing operators $\widetilde K_i$ will also be constant, so that the corresponding quasilinear system $u_{t_i} = \widetilde K_i u_x$, $i=1,\dots, n$ will be, in fact, linear.  This leads us to the following conclusion.   

\begin{Corollary}\label{c2}
Let $\Sym L$ be the symmetry algebra of $\gl$-regular algebraically generic Nijenhuis operator $L$.  Consider an arbitrary basis $M_1,\dots, M_n$ of $\Sym L$ and the corresponding collection $K_1=\Id, K_2, \dots, K_n$ of Killing $(1,1)$-tensors constructed from this bases via  \eqref{eq:hiMi}, \eqref{eq:hiKi} (see Theorem  \ref{t4}). Then the quasilinear system 
$$
u_{t_i} = K_i u_x, \quad i=1,\dots, n,  \  x=t_1,
$$
can be linearised by a suitable reciprocal transform.
\end{Corollary}

\subsection*{Acknowledgments.}  A.\,B. and A.\,K. were supported by the Ministry of Science and Higher Education of the Republic of Kazakhstan (grant No. AP23483476), and   V.\,M. by  the DFG (projects 455806247 and 529233771) and the ARC Discovery Programme DP210100951.

  Data sharing is not applicable to this article, as no datasets were generated or analysed during the current study. The authors declare no conflicts of interest.

\printbibliography

\end{document}